\documentclass[conference,letterpaper]{IEEEtran}

\usepackage[utf8]{inputenc} 
\usepackage[T1]{fontenc}
\usepackage{url}
\usepackage{ifthen,color}

\usepackage{cite}
\usepackage[cmex10]{amsmath} % Use the [cmex10] option to ensure complicance
\usepackage{amsthm}

\newtheorem{lemma}{Lemma}
\newtheorem{theorem}{Theorem}
\newtheorem{remark}{Remark}
\newtheorem{corollary}{Corollary}
\newcommand{\pnext}{.\textit{nx}}
\newcommand{\pprev}{.\textit{pr}}
\usepackage{caption}
\usepackage{subcaption}
\usepackage{graphicx}
\newcommand{\msf}{\mathsf}

\input{widebar_hack.tex}

\begin{document}
\title{On optimal relay placement in directional networks}

% %%% Single author, or several authors with same affiliation:
% \author{%
%   \IEEEauthorblockN{Stefan M.~Moser}
%   \IEEEauthorblockA{ETH Zürich\\
%                     ISI (D-ITET)\\
%                     CH-8092 Zürich, Switzerland\\
%                     Email: moser@isi.ee.ethz.ch}
% }

%%% Several authors with up to three affiliations:
\author{
\IEEEauthorblockN{Mine Gokce Dogan, Yahya H. Ezzeldin, Christina Fragouli}
\IEEEauthorblockA{UCLA, Los Angeles, CA 90095, USA\\
Email: \{minedogan96, yahya.ezzeldin, christina.fragouli\}@ucla.edu}
}

%%% Many authors with many affiliations:
% \author{%
%   \IEEEauthorblockN{Albus Dumbledore\IEEEauthorrefmark{1},
%                     Olympe Maxime\IEEEauthorrefmark{2},
%                     Stefan M.~Moser\IEEEauthorrefmark{3}\IEEEauthorrefmark{4},
%                     and Harry Potter\IEEEauthorrefmark{1}}
%   \IEEEauthorblockA{\IEEEauthorrefmark{1}%
%                     Hogwarts School of Witchcraft and Wizardry,
%                     1714 Hogsmeade, Scotland,
%                     \{dumbledore, potter\}@hogwarts.edu}
%   \IEEEauthorblockA{\IEEEauthorrefmark{2}%
%                     Beauxbatons Academy of Magic,
%                     1290 Pyrénées, France,
%                     maxime@beauxbatons.edu}
%   \IEEEauthorblockA{\IEEEauthorrefmark{3}%
%                     ETH Zürich, ISI (D-ITET), ETH Zentrum, 
%                     CH-8092 Zürich, Switzerland,
%                     moser@isi.ee.ethz.ch}
%   \IEEEauthorblockA{\IEEEauthorrefmark{4}%
%                     National Chiao Tung University (NCTU), 
%                     Hsinchu, Taiwan,
%                     moser@isi.ee.ethz.ch}
% }

\maketitle

%%%%%%
%% Abstract: 
%% If your paper is eligible for the student paper award, please add
%% the comment "THIS PAPER IS ELIGIBLE FOR THE STUDENT PAPER
%% AWARD." as a first line in the abstract. 
%% For the final version of the accepted paper, please do not forget
%% to remove this comment!
%%
\begin{abstract}
In this paper, we study the problem of optimal topology design in wireless networks equipped with highly-directional transmission antennas. We use the 1-2-1 network model to characterize the optimal placement of two relays that assist the communication between a source-destination pair. We analytically show that under some conditions on the distance between the source-destination pair, the optimal topology in terms of maximizing the network throughput is to place the relays as close as possible to the source and the destination.
\end{abstract}

\section{Introduction}
We consider a source that would like to transmit information to a destination with the help of $N$ relays,  using directional communication. A typical network information theory question is, given a fixed network topology, what is the maximum rate that the source can communicate to the destination (network capacity)? We here look at the reverse question: assuming that we have the freedom of arbitrarily placing the relays,  where should we place them, so that the capacity is maximized?

We focus this work on directional communication networks, 
which have significant impact in next generation systems.
Indeed, as we move towards using higher and higher frequencies,
communication increasingly becomes directional, to combat the severe path-loss and improve the data rates. A prominent example is millimeter-wave (mmWave) networks, that  are anticipated to play an integral role for the fifth-generation (5G) cellular systems due to the large bandwidth they provide.

Understanding how to optimally place relays in such networks is a timely question. There are currently a number of initiated projects that  aim to deploy infrastructure for directional communications, such as Terragraph~\cite{terragraph}. Moreover, a number of applications envisage creating mobile backbones, such as flexible UAV-assisted wireless networks~\cite{survey}, where adapting the topology has low cost. 

Finding the optimal topology is not a straightforward question to answer. 
The main challenge is the ability of beam steering, which makes the capacity calculation dependent on the schedule (which node transmits to which node and for how long). We note that the optimal schedule depends on the underlying configuration; thus for every network configuration, one would need to calculate the associated optimal schedule to find the capacity. 

On the other hand, finding the optimal topology can be worth the effort - randomly placing relays, as  
 Fig.~\ref{fig1} depicts, can result in underwhelming performance.
\begin{figure}[t]
\centerline{\includegraphics[width=0.72\columnwidth]{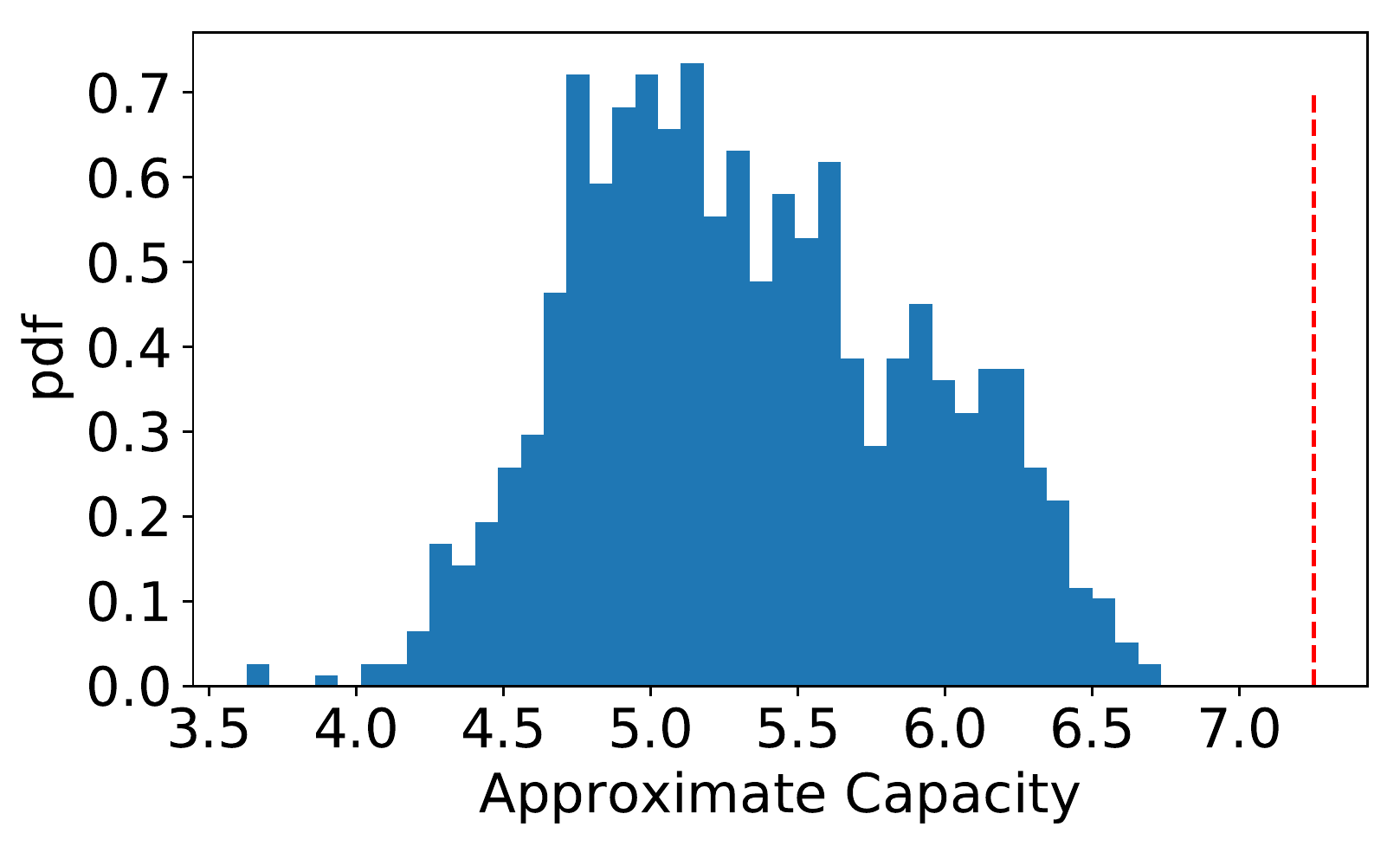}}
\caption{Empirical distribution of the approximate capacity for relays placed uniformly at random in space. The red line marks the approximate capacity of the optimal topology.}
\vspace{-1em}
\label{fig1}
\end{figure}
This figure assumes that two relays are placed uniformly at random in the space between a fixed source and destination location where the distance between the source and the destination is $d = 600\sqrt{2}$. The distribution shown in the figure was computed over 1000 random topologies with the red line indicating the performance of the optimal relay placement.
We find that the (approximate) capacity  can vary significantly depending on where the two relays are placed, and can be much lower (on average almost $50\%$ lower) than the capacity of the optimal configuration.

Our approach to address this problem builds on the 1-2-1 network model which was introduced in~\cite{ezzeldin} to study the capacity of wireless networks with steerable highly-directional antennas. In particular, in~\cite{ezzeldin}, it was proved that for a predetermined topology, the capacity of the network can be approximated to within a universal constant gap\footnote{Constant gap refers to a quantity that is independent of the channel coefficients and operating SNR, and only depends on the number of nodes.} and its optimal beam schedule can be found in polynomial-time.

Our main result in this paper is to characterize the optimal topology in terms of the approximate capacity for a two relay Gaussian Full-Duplex (FD) 1-2-1 network under a path-loss propagation model. Surprisingly, we prove that when the distance between the source and destination allows reasonable point-to-point capacity between the two nodes using directional antennas, the optimal topology concentrates relays at the source and destination positions, such that one relay is as close as possible to the source and the other to the destination. 
This understanding of the optimal topology for two relays, offers a first step towards creating optimal topologies for larger directional communication networks.

\noindent{\bf Related Work.} Several works in the literature focus on the optimization of relay locations in wireless networks that employ omnidirectional antennas. Optimal relay placement is studied within a cellular system in~\cite{WangSHCC08}. The authors in~\cite{XuZLW11} and~\cite{JointPower} characterize the optimal placement of a single amplify-forward relay in a cooperative communication network with and without joint optimization of power allocation. In \cite{TanoliSKNHAKSA20}, the authors find the best relay location over a finite set of possible locations for bi-directional transmission. However, they do not optimize the relay locations in a continuous domain and their work does not include the scheduling aspects of directional networks. In \cite{deterministic}, the optimal placement of relay nodes is studied for the linear deterministic network model for a wireless network. In contrast, our work considers the relay placement problem in networks with steerable directional antennas and does not put restrictions on the schemes that are employed by the relays. 

Perhaps the closest to our work are the results on relay placement in mmWave networks~\cite{robust,KongYWTCV17}. These approaches only consider picking the best topology among a class of predetermined topologies and beamforming schedules or present heuristic metrics based on link qualities in the network. Therefore, none of these works identifies the optimal topology structure as we do. Our study considers the unicast approximate capacity as defined by the 1-2-1 information-theoretic model and allows a more extensive topology search with each topology operating with its optimal beam schedule.

\noindent \textbf{Paper Organization.} Section~\ref{sec:model} provides background on the 1-2-1 network model for mmWave networks and introduces our channel model. Section~\ref{sec:results} introduces our main theorem, that is proved in Section~\ref{theorem_proof}. Section~\ref{sec:concl} concludes the paper.

\section{System Model and Background}\label{sec:model}
%\subsection{Gaussian 1-2-1 Networks}
We consider the Gaussian Full-Duplex (FD) 1-2-1 model which was proposed in~\cite{ezzeldin} to study the information-theoretic capacity of multi-hop wireless networks that utilize directional transmissions for communication. In an $N$-relay Gaussian FD 1-2-1 network, $N$ relays assist the communication between a source node (node 0) and a destination node (node $N+1$). 
Each node in the network can transmit and receive simultaneously and is equipped with a single highly-directional transmit beam and a single highly-directional receive beam. 
At any particular instance, a node in a 1-2-1 network can only transmit to at most one node and receive from at most one node by directing its transmit and receive beam, respectively.\\
In order for two nodes to communicate, they need to activate the link between them by steering their beams towards each other (thus, the name of the 1-2-1 model). In the following, we discuss capacity results for FD 1-2-1 networks.

\noindent {\bf Capacity of FD 1-2-1 networks.} In~\cite{ezzeldin}, it was shown that the capacity of an $N$-relay Gaussian FD 1-2-1 network can be approximated to within a constant gap that only depends on the number of nodes in the network\footnote{The constant gap is due to the fact that the choice of the optimal schedule can potentially be used to relay information of the transmitted message. In case the schedule cannot be used to convey information about the message, the characterization is the exact capacity of the network.}. In particular, it was shown in~\cite{ezzeldin} that the approximate capacity and its optimal schedule can be computed in polynomial-time through the following Linear Program (LP),
%\begin{equation}\label{capacity_lp}
%\begin{array}{rl}
%\displaystyle
%\max_{\mathbf{\lambda},\mathbf{F}} & \sum\limits_{j=0}^N F_{d,j}\\
%\text{subject to} & 0 \leq F_{j,i} \leq \lambda_{\ell_{j,i}}\ell_{j,i} \qquad \forall (i,j) \in [0:N] \times [1:N+1]\\ & \sum\limits_{j \in [1:N+1]\backslash \{i\}} F_{j,i} = \sum\limits_{k \in [0:N]\backslash \{i\}} F_{i,k} \qquad \forall i \in [1:N] \\ & \sum\limits_{j \in [1:N+1]\backslash \{i\}} \lambda_{\ell_{j,i}} \leq 1 \qquad \forall i \in [0:N]\\&  \sum\limits_{i \in [0:N]\backslash \{j\}} \lambda_{\ell_{j,i}} \leq 1 \qquad \forall j \in [1:N+1] \\&  \lambda_{\ell_{j,i}} \geq 0 \qquad \forall (i,j) \in [0:N] \times [1:N+1]
%\end{array}
%\end{equation}
%
\begin{align}\label{capacity_lp}
%\begin{align}
% \begin{array}{llll}
& \ \rm{P1:}\ \widebar{\msf{C}} = \displaystyle\max\limits_{\lambda,F} \sum_{j=0}^N F_{N+1,j}, & & \nonumber \\
&({\rm P1}a) \ 0 \leq F_{j,i} \leq \lambda_{\ell_{j,i}}\ell_{j,i},& \forall (i,j) {\in} [0\!:\!N] \!\times\! [1{:}N{+}1], & \nonumber\\
&({\rm P1}b) \hspace{-0.2in}\displaystyle \sum_{\substack{j \in [1:N{+}1]\backslash \{i\}}} \hspace{-0.24in} F_{j,i} = \!\!\!\!\!\!\! \displaystyle \sum_{\substack{k \in [0{:}N]\backslash \{i\}}} \hspace{-0.24in} F_{i,k}, & \forall i \in [1{:}N], & \nonumber\\
&({\rm P1}c) \! \! \! \! \displaystyle \sum_{\substack{j \in [1:N{+}1]\backslash\{i\}}} \! \! \! \! \!\! \lambda_{\ell_{j,i}} \leq 1, & \forall i \in [0{:}N], & \\
&({\rm P1}d) \! \! \! \! \displaystyle \sum_{\substack{i \in [0:N]\backslash\{j\}}} \! \! \! \!  \lambda_{\ell_{j,i}} \leq 1, & \forall j \in [1{:}N+1], & \nonumber\\
&({\rm P1}e) \ \lambda_{\ell_{j,i}} \geq 0,  &\forall (i,j)  {\in} [0\!:\!N] \!\times\! [1\!:\!N{+}1] \nonumber,
% \end{array}
%\end{align}
\end{align}
where: (i) $\widebar{\msf C}$ is the approximate network capacity; (ii) $\ell_{j,i}$ denotes the point-to-point link capacity from node $i$ to node $j$ when the respective transmit and receive beams are aligned; (iii) $F_{j,i}$ represents the amount of information flowing from node $i$ to node $j$; and (iv) $\lambda_{\ell_{j,i}}$ denotes the fraction of time for which the link of capacity $\ell_{j,i}$ is active, based on the schedule used to operate/align the antenna beams in the network.

% For a fixed schedule, i.e, fixed $\lambda_{\ell_{j,i}}$, the linear program in \eqref{capacity_lp} turns into the maximum flow problem and it is equivalent to the minimum cut problem. Through using $\lambda_{\ell_{j,i}}$ as a variable, the program maximizes the minimum cut in the network, i.e., maximizes the approximate capacity.

\noindent{\bf Cut-set formulation of approximate capacity.}
We observe that for a fixed feasible configuration of link activation times $\{\lambda_{\ell_{ji}}\}$, the LP in~\eqref{capacity_lp} is the classical max-flow problem over a graph with edge capacities $\lambda_{\ell_{j,i}}\ell_{j,i}$. Thus, we can rewrite~\eqref{capacity_lp} by replacing the max-flow LP with its cut-set dual problem. In particular, we get the following LP,
\begin{align}\label{capacity_lp_cut}
%\begin{align}
% \begin{array}{llll}
& \ \rm{P2:}\ \widebar{\msf{C}} = \max_{\lambda, \alpha} \alpha \nonumber \\
&({\rm P2}a) \! \! \! \! \displaystyle \sum_{\substack{j \in [1:N{+}1]\backslash\{i\}}} \! \! \! \! \!\! \lambda_{\ell_{j,i}} \leq 1,  \hspace{1.05in} \forall i \in [0{:}N], \nonumber\\
&({\rm P2}b) \! \! \! \! \displaystyle \sum_{\substack{i \in [0:N]\backslash\{j\}}} \! \! \! \!  \lambda_{\ell_{j,i}} \leq 1,  \hspace{0.9in} \forall j \in [1{:}N+1], \\
&({\rm P2}c) \ \lambda_{\ell_{j,i}} \geq 0, \hspace{0.73in} \forall (i,j)  {\in} [0\!:\!N] \!\times\! [1\!:\!N{+}1] \nonumber,\\
&({\rm P2}d) \alpha \!\leq\!\!\!\!\!\!\!\sum_{\substack{(i,j): \\i \in \Omega, j \in \Omega^c}} \hspace{-0.18in}\lambda_{\ell_{j,i}}\ell_{j,i},\ \ \forall \Omega\!\subseteq\![0{:}N\!+\!1], 0 {\in} \Omega, N\!\!+\!1 {\in}\Omega^c,  \nonumber
% \end{array}
%\end{align}
\end{align}
where for each constraint $({\rm P2}d)$, the right-hand side (RHS) represents the cut value for that particular cut $\Omega$ given the link activations $\{\lambda_{\ell_{j,i}}\}$. Thus, the LP P2 maximizes the minimum cut by optimizing the link activation times that satisfy the constraints $({\rm P2}a-c)$.

If we specialize the LP P2 to the case where there are $N=2$ relays in a network, we get that
\begin{align}\label{capacity_lp_cut_N_2}
%\begin{align}
% \begin{array}{llll}
& \ \rm{P3:}\ \widebar{\msf{C}} = \max_{\lambda, \alpha} \alpha  \nonumber \\
&({\rm P3}a) \! \! \! \! \displaystyle \sum_{\substack{j \in [1:3]\backslash\{i\}}} \! \! \! \! \!\! \lambda_{\ell_{j,i}} \leq 1, \qquad\qquad\qquad\  \forall i \in [0{:}2],  \nonumber\\
&({\rm P3}b) \! \! \! \! \displaystyle \sum_{\substack{i \in [0:2]\backslash\{j\}}} \! \! \! \!  \lambda_{\ell_{j,i}} \leq 1, \qquad\qquad\qquad \forall j \in [1{:}3], \\
&({\rm P3}c) \ \lambda_{\ell_{j,i}} \geq 0,  \qquad\qquad\  \forall (i,j)  {\in} [0\!:\!2] \!\times\! [1\!:\!3] \nonumber,\\
&({\rm P3}d) \ \alpha \leq \lambda_{\ell_{3,1}}\ell_{3,1}+\lambda_{\ell_{3,2}}\ell_{3,2}+\lambda_{\ell_{3,0}}\ell_{3,0},  \nonumber \\
&({\rm P3}e) \ \alpha \leq \lambda_{\ell_{1,0}}\ell_{1,0}+\lambda_{\ell_{2,0}}\ell_{2,0}+\lambda_{\ell_{3,0}}\ell_{3,0},  \nonumber \\
&({\rm P3}f) \ \alpha \leq \lambda_{\ell_{2,1}}\ell_{2,1}+\lambda_{\ell_{3,1}}\ell_{3,1}+\lambda_{\ell_{2,0}}\ell_{2,0}+\lambda_{\ell_{3,0}}\ell_{3,0},  \nonumber\\
&({\rm P3}g) \ \alpha \leq \lambda_{\ell_{1,0}}\ell_{1,0}+\lambda_{\ell_{3,2}}\ell_{3,2}+ \lambda_{\ell_{1,2}}\ell_{1,2}+\lambda_{\ell_{3,0}}\ell_{3,0}.  \nonumber
% \end{array}
%\end{align}
\end{align}

Our main results in this paper for Gaussian 1-2-1 networks with $N=2$ relays, which are presented in the following section, will rely heavily on the cut-set formulation for the approximate capacity shown in~\eqref{capacity_lp_cut_N_2}.

\noindent {\bf Network topology and propagation model.}
Throughout the remainder of the paper, we consider a Gaussian 1-2-1 network with $N=2$ relays as illustrated in Fig.~\ref{arbitrary_network}. 
\begin{figure}[htbp]
\centerline{\includegraphics[width=0.7\columnwidth]{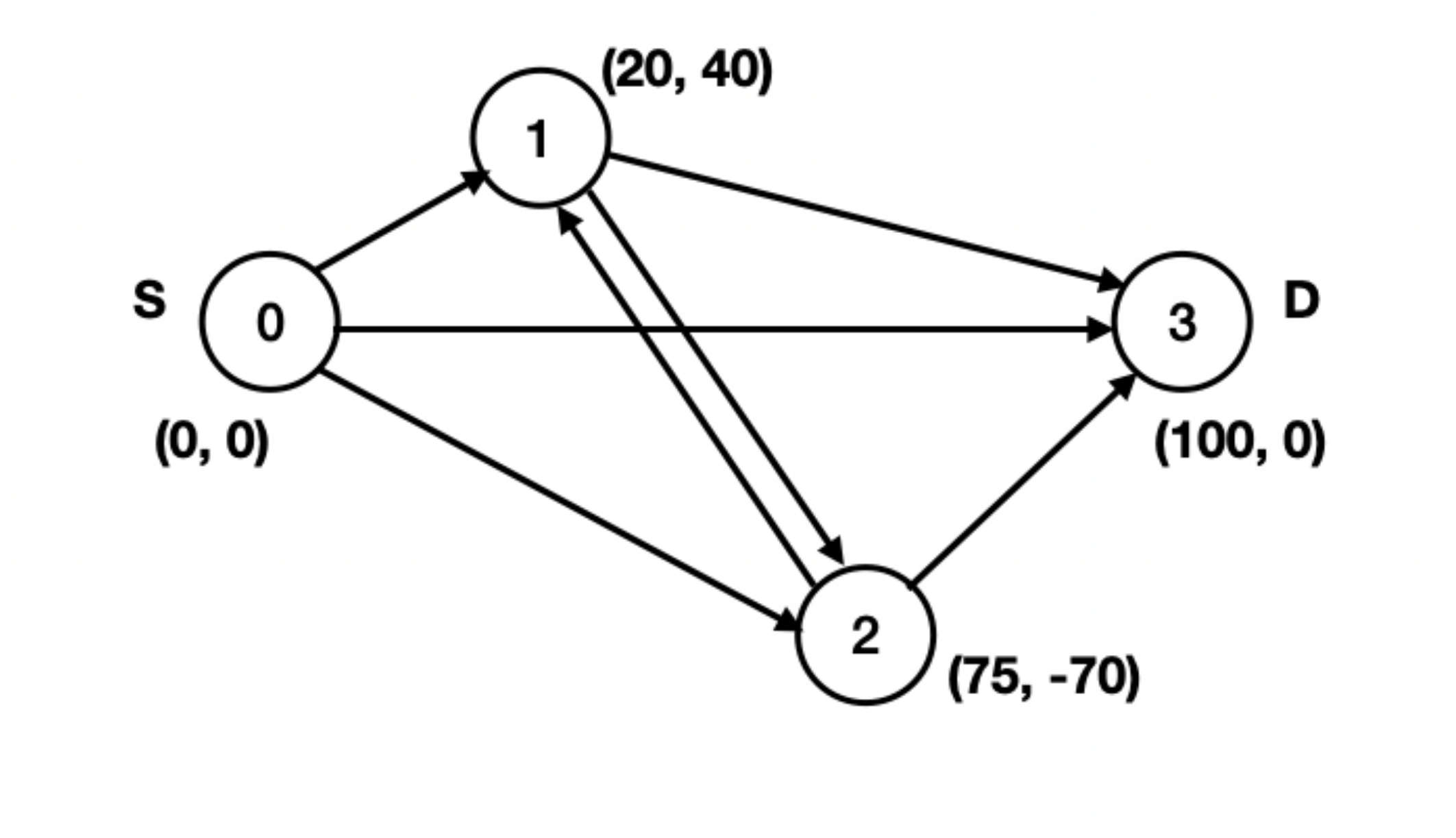}}
\caption{An example of the considered network topology.}
\label{arbitrary_network}
\end{figure}
Without loss of generality, we assume that the source is at coordinates $(0,0)$, the destination is at $(d,0)$ and the two relays are at $(x_1,y_1)$ and $(x_2,y_2)$, respectively. The x- coordinates $x_1$ and $x_2$, and y-coordinates $y_1$ and $y_2$ can be arbitrarily different. 

For the point-to-point link capacity between any two nodes with aligned transmit/receive beams, we use a path-loss model while assuming that all nodes follow a uniform transmission power constraint. In particular, we have that 
% are subject to a uniform power constraint $P$ and that link capacities decay with follow LOS (line-of-sight) channel model for the link capacities $\ell_{j,i}$.
\begin{equation}\label{channel_model}
    \ell_{j,i} = \log \left(1+\frac{\gamma}{d_{j,i}^a}\right) \stackrel{(a)}{\approx} \log \left(\frac{\gamma}{d_{j,i}^a}\right)
\end{equation}
where: (i) $d_{j,i}$ is the distance between node $i$ and node $j$; (ii) $a > 1$ is the path loss exponent; (iii) the arbitrary parameter $\gamma > 0$ subsumes the effects of the transmission power, the aligned transmit and receive antenna beams and communication wavelength. The approximation in $(a)$ assumes that $\gamma/d_{j,i}^a \gg 1, \forall (i,j) \in  [0:N]\times[1:N+1]$.
\begin{remark}
{\rm 
Note that from~\eqref{channel_model}, it follows that 
\begin{align*}
\frac{\gamma}{d_{j,i}^a} > 1 \ \ \implies \ \ \ell_{j,i} = \log \left(1+\frac{\gamma}{d_{j,i}^a}\right)\leq \log \left(\frac{\gamma}{d_{j,i}^a}\right) + 1.
\end{align*}
As a result, it is not difficult to see that working with the approximation of $\ell_{j,i}$ in computing the approximate capacity using LP P1 in~\eqref{capacity_lp}, can at most reduce the approximate capacity by a single bit. Thus, with a slight abuse of notation, we take $\ell_{j,i} = \log \left(\gamma / d_{j,i}^a\right)$ for the remainder of the paper.
}  
\end{remark}
In the following section, we present our key result in this paper that characterizes the best topology among the class of 2-relay Gaussian 1-2-1 networks described above.

\section{Optimal Topology for the Two Relay Network}\label{sec:results}
In this section, we answer  the question: what is the optimal placement of two relays in a Gaussian 1-2-1 FD network that maximizes the rate transferred between a source and a destination? Our key result is summarized by the following theorem which is proved in Section~\ref{theorem_proof}.
\begin{theorem}\label{main_theorem}
\textit{Consider a $2$-relay Gaussian FD 1-2-1 network with a topology and propagation model as described in Section~\ref{sec:model}. For distance $d$ such that $\frac{\gamma}{d^a} > 3^a$, the optimal topology in terms of maximizing the approximate capacity is to place one of the relays as close to the source as possible and the other relay as close to the destination as possible.}
\end{theorem}

A careful computation of the optimal solution of LP P3 in~\eqref{capacity_lp_cut_N_2} when the two relays approach the source and destination, respectively, gives the following corollary.

\begin{corollary}\label{optimum_cap}
For the 2-relay Gaussian 1-2-1 network with conditions satisfying Theorem~\ref{main_theorem}, the best approximate capacity achievable by any topology is given by
\begin{equation}
\begin{aligned}
\widebar{\msf{C}}^\star = 2\log\left(\frac{\gamma}{d^a}\right).
\end{aligned}
\end{equation}
\end{corollary}
\begin{proof}
Setting the two relays at the source and destination positions, respectively, the link capacities $\ell_{10}$ and $\ell_{32}$ in LP P3 in~\eqref{capacity_lp_cut_N_2} grow infinitely large. 
Thus, only the constraint $(\rm{P3}f)$ remains as an upper bound on $\alpha$. Additionally, we have $\ell_{20} = \ell_{21} = \ell_{31} = \ell_{30} = \log\left(\gamma/d^a\right)$. Setting two non-conflicting $\lambda$'s to unity (for example $\lambda_{20} = \lambda_{31} = 1$) maximizes the RHS of $(\rm{P3}f)$ and proves the corollary.
\end{proof}
\begin{remark}
{\rm 
The results in Theorem~\ref{main_theorem} and Corollary~\ref{optimum_cap} suggest that in the absence of potential blockage between the source and destination, the optimal relay placement is for the network to concentrate two nodes (source and one relay) at the source location, and similarly at the destination location, in order to virtually increase the number of transmit and receive beams at the disposal of the source and destination, respectively. Our simulation in Fig.~\ref{varying_beta} also suggests that, for $d = 200\sqrt{2}$ (it satisfies the condition in Theorem~\ref{main_theorem}), the approximate capacity increases as $\beta$ decreases (the relays approach the source and the destination). We note that the reason for the change of the slope in Fig.~\ref{varying_beta} is due to the change in the optimal schedule. When $\beta > 1/3$, the optimal schedule becomes a routing schedule, i.e., the source sends all the information to relay $1$, relay $1$ sends it to relay $2$ and relay $2$ sends it to the destination.
}
\end{remark}
\begin{figure}[t]
\centerline{\includegraphics[width=0.6\columnwidth]{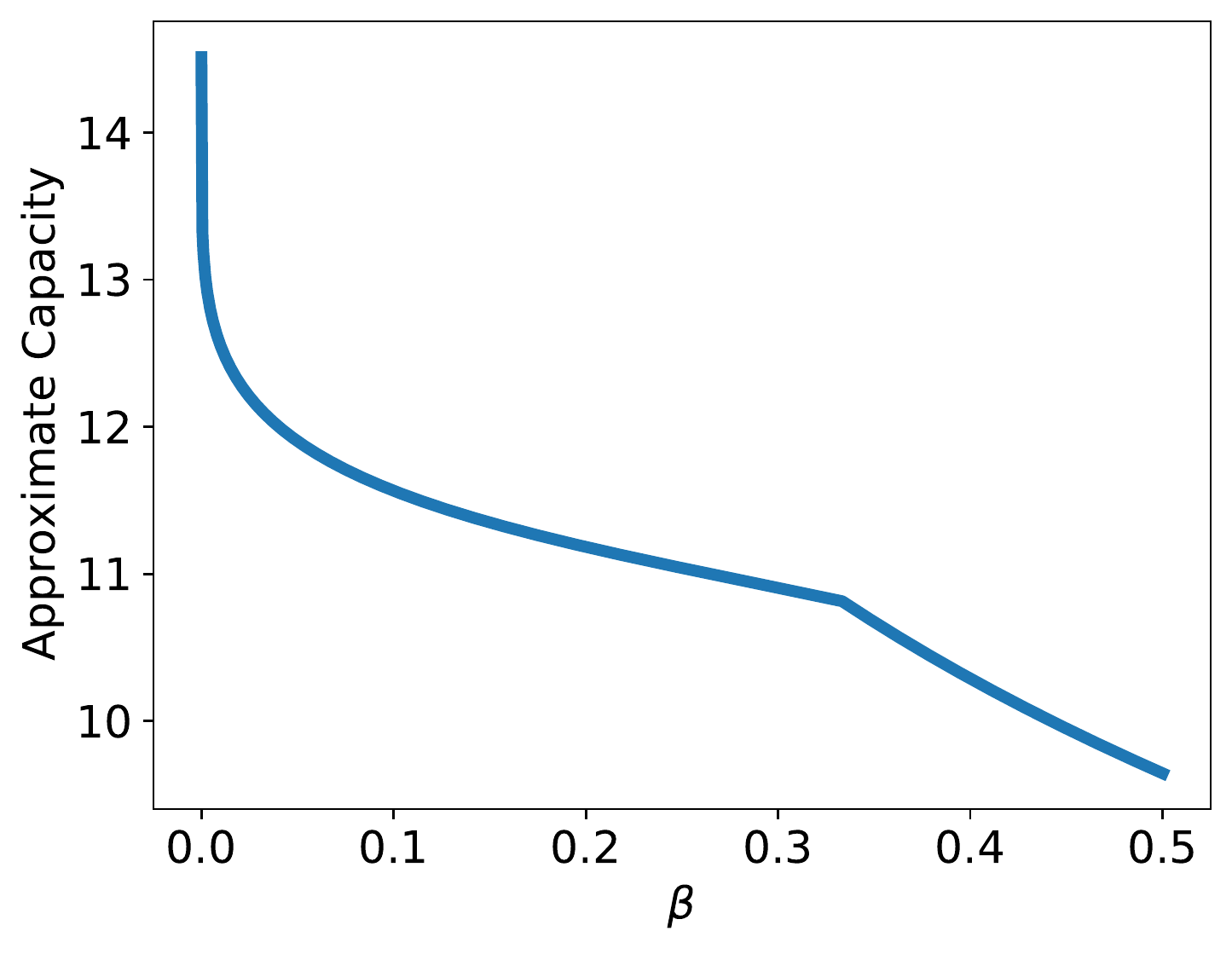}}
\caption{Change in the approximate capacity with respect to $\beta$, for $d = 200\sqrt{2}$ and $a = 2$.}
\label{varying_beta}
\end{figure}
% \textbf{\begin{remark}
% {\rm 
% The topology suggested by Theorem~\ref{main_theorem} for two relays is similar to an earlier result on optimal relay placement for linear deterministic networks in~\cite{deterministic}, where node concentration at the source and destination provided the highest capacity over all classes of network topologies. Contrary to our studied model, the result~\cite{deterministic} over deterministic networks and do not consider the potential problem complexity due to scheduling of transmissions. 
% }
% \end{remark}}

\begin{remark}
{\rm 
Intuitively, one would expect that as the distance between the source and destination increases, the relays would try to form a route (line network) between the source and destination in order to mitigate the effect of path-loss. Theorem~\ref{main_theorem} can be viewed as a sufficient condition on the distance for which such an approach is no longer the optimal. In fact it is not difficult to show analytically that for $d \geq \gamma^{1/a} / 3$, the single route (line network) with $d_{10} = d_{21} = d_{32} = d/3$ shown in Fig.~\ref{line_network} outperforms the topology suggested by Theorem~\ref{main_theorem}.
Our evaluations at different distances, shown in Fig.~\ref{comparison}, also suggest that this critical bound on $d$ holds even when the approximate capacity is found without the approximation in~\eqref{channel_model}. 
}
\end{remark}
%In Theorem \ref{main_theorem}, we assume that there is no blockage or interference between the nodes and $\frac{\gamma}{d_{j,i}^a} \gg 1 \quad \forall (i,j) \in [0:N]\times[1:N+1]$ as stated in the system model. 
%\begin{proof}
%The proof of Theorem \ref{main_theorem} is delegated to the next subsection.
%\end{proof} 
% Before discussing the proof of Theorem \ref{main_theorem}, we first provide an overview of high-level steps in the proof.\\
\begin{figure}[t]
\centerline{\includegraphics[width=0.65\columnwidth]{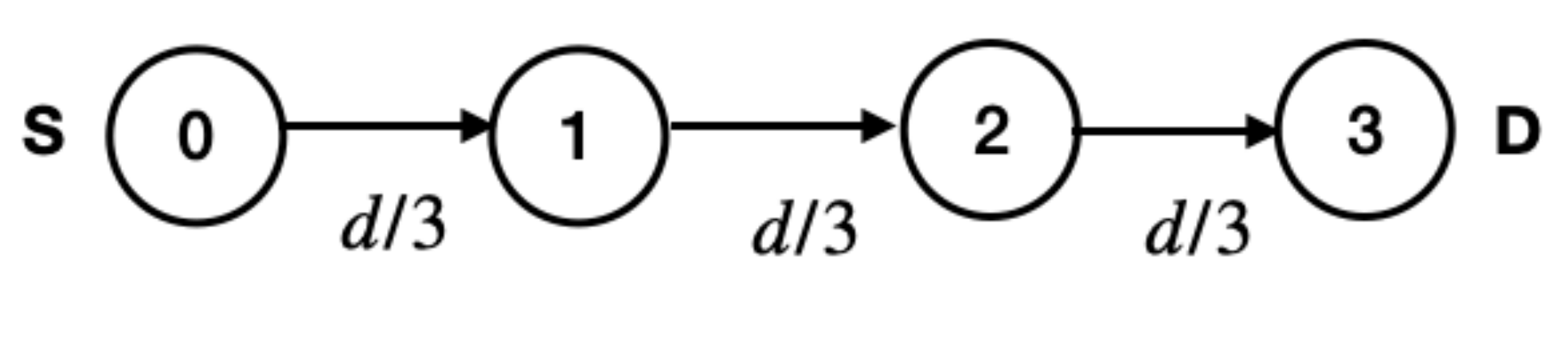}}
\caption{The line network with equal spacing.}
\vspace{-1em}
\label{line_network}
\end{figure}
\begin{figure}[t]
\centerline{\includegraphics[width=0.7\columnwidth]{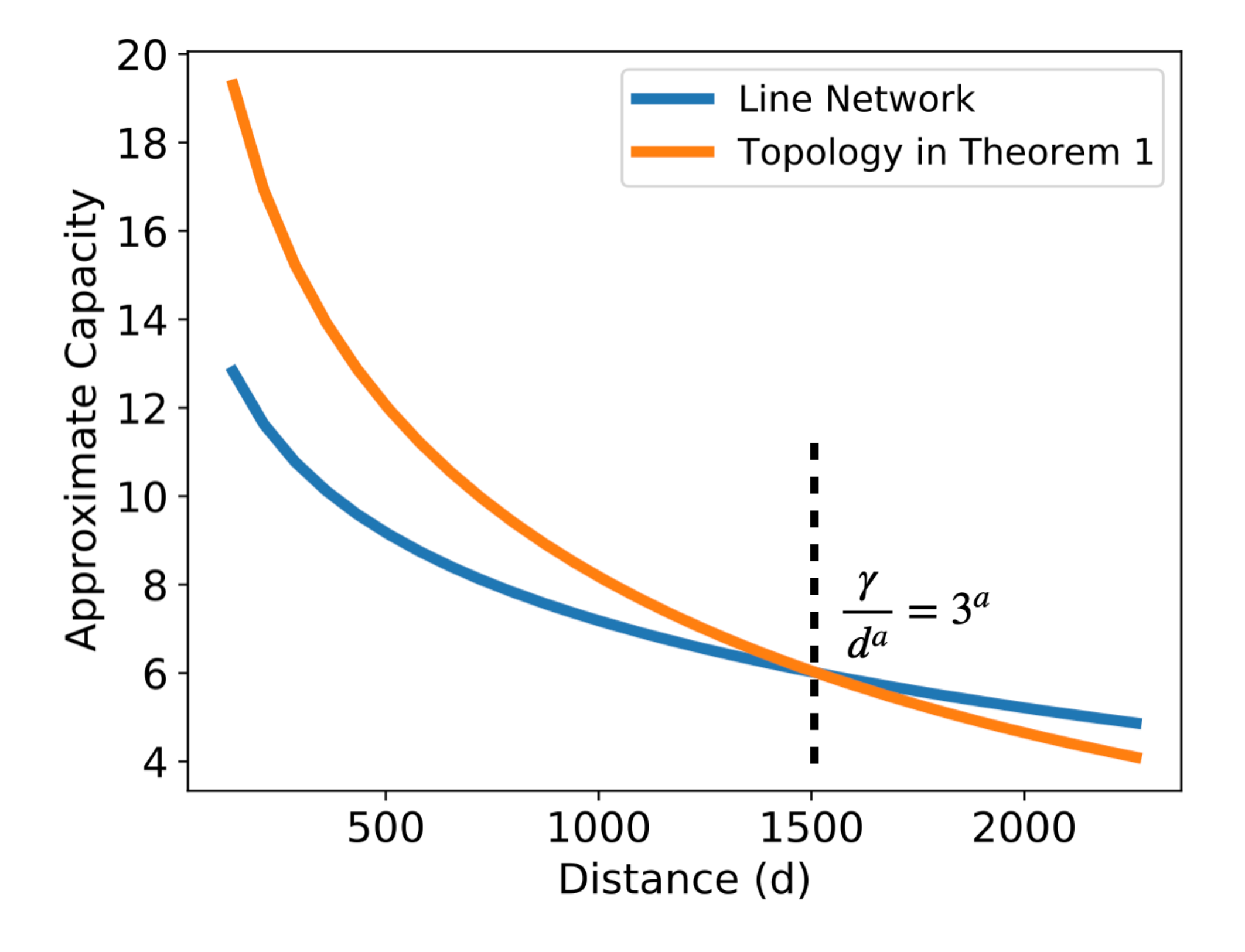}}
\caption{Comparison between the approximate capacity of the topology proposed in Theorem~\ref{main_theorem} and the line network with equal spacing between nodes when $a = 2$. 
}
\label{comparison}
\end{figure}

%=========================================
\section{Proof of Theorem~\ref{main_theorem}}\label{theorem_proof}
We start this section by giving an outline of the key steps in the proof of Theorem~\ref{main_theorem}, before delving into the details of the proof in the following subsections.
\subsection{Proof Outline}
\noindent{\bf [1. Projected Topologies]} The first step in proving Theorem~\ref{main_theorem} is to observe that for any network topology, projecting the topology on the axis connecting the source and destination, can only improve the approximate capacity since the distance between any two nodes in the network does not increase and as a result the approximate capacity from LP P3 cannot decrease. Thus in the remainder of the proof, we will focus our attention on the network topologies where the y-coordinates of the two relays are set to zero. For the example topology shown in Fig.~\ref{arbitrary_network}, its projected counterpart is shown in Fig.~\ref{projected_network}. 
\begin{figure}[htbp]
\centerline{\includegraphics[width=0.7\columnwidth]{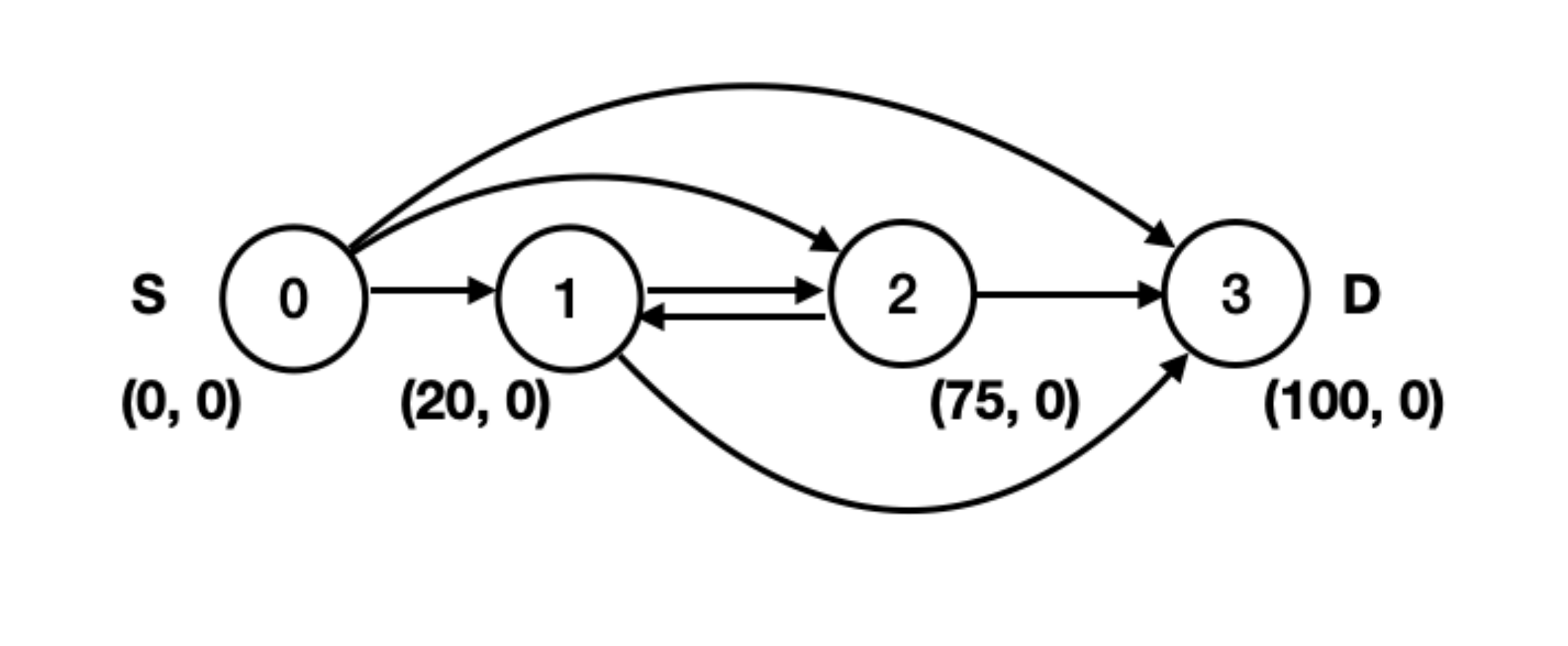}}
\caption{Projected topology for the topology shown in Fig.~\ref{arbitrary_network}.}
\label{projected_network}
\end{figure}
To ease the representation of such topologies, we define the distances between the nodes in the network using fractions $\beta_1$ and $\beta_2$ of the distance $d$ between the source and destination where $\beta_1,~\beta_2 \in [0,1]$. In particular, the two relays are assumed to be located at the coordinates $(\beta_1d,0)$ and $(d(1-\beta_2),0)$, respectively where $\beta_1d \leq d-\beta_2d$. For the example topology in Fig.~\ref{projected_network}, $d= 100$, $\beta_1d = 20$ and $\beta_2d = 25$.

\medskip

\noindent{\bf [2. Symmetric Networks]}
The next step in proving Theorem~\ref{main_theorem} is to leverage the following lemma. 
\begin{lemma}\label{symmetric_network_lemma}
For a projected asymmetric network as in Fig.~\ref{projected_network}, we can always find a projected symmetric network ($\beta_1' = \beta_2' = \beta$) that achieves at least the same approximate capacity.
\end{lemma} 
\begin{proof}
The proof is delegated to Appendix \ref{asymm_symm}.
\end{proof}
Therefore in the remainder of the proof of Theorem~\ref{main_theorem}, we will focus on projected symmetric network topologies where we use fraction $\beta \in [0, 1/2]$ of the distance $d$ to define the distances between nodes as follows.
\begin{equation}\label{modified_distances}
\begin{aligned}
    d_{1,0} = d_{3,2} &= \beta d,\\
    d_{2,0} = d_{3,1} &= (1-\beta)d,\\
    d_{1,2} = d_{2,1} &= (1-2\beta)d,\\
    d_{3,0} &= d.
\end{aligned}
\end{equation}
From this, it follows that the link capacities can be written as follows by using the channel model in \eqref{channel_model}.
\begin{equation}\label{modified_linkcap}
\begin{aligned}
    \ell_{1,0} = \ell_{3,2} = \ell_1 &= \log\left(\frac{\gamma}{\beta^ad^a}\right),\\
    \ell_{2,0} = \ell_{3,1} = \ell_2 &= \log\left(\frac{\gamma}{(1-\beta)^ad^a}\right),\\
    \ell_{1,2} = \ell_{2,1} = \ell_3 &= \log\left(\frac{\gamma}{(1-2\beta)^ad^a}\right),\\
    \ell_{3,0} = \ell_4 &= \log\left(\frac{\gamma}{d^a}\right).\\
\end{aligned}
\end{equation}

\medskip

\noindent{\bf [3. Properties of optimal solution for each topology]} Next, we prove that for the class of projected symmetric topologies described in the previous step, there always exists an optimal solution for the LP P3 in~\eqref{capacity_lp_cut_N_2} where: $(i)$ the link activation times are symmetric; and $(ii)$ the source and destination do not need to directly communicate. This is presented by the following lemma. 

\begin{lemma}\label{properties}
Consider a 2-relay Gaussian FD 1-2-1 network, with source and destination at $(0,0)$ and $(d,0)$, respectively. Assume that for some $\beta \in [0,1/2]$, the two relays are located at the coordinates $(\beta d, 0)$ and $((1-\beta)d, 0)$, respectively, then there exists an optimal solution for LP {\rm P3} that satisfies the following properties:
% for the so have the o as in Figure \ref{projected_network} where $x_1  = \beta d$ and $x_2 = d-\beta d$, we have the following properties.
\begin{align}
&\label{property1} (a) \quad 
\begin{aligned}
\lambda_{\ell_{1,0}} = \lambda_{\ell_{3,2}} = \lambda_1,\\
\lambda_{\ell_{2,0}} = \lambda_{\ell_{3,1}} = \lambda_2,
\end{aligned} \\
% \end{equation}
% \begin{equation}
&\label{property2} (b) \quad 
\begin{aligned}
\lambda_{\ell_{1,2}} = 0,\\
\lambda_2 + \lambda_{\ell_{2,1}} = 1,
\end{aligned} \\
% \end{equation}
% \begin{equation}
&\label{property3} (c) \quad 
\begin{aligned}
\lambda_{\ell_{3,0}} = 0,\\
\lambda_{1} + \lambda_{2} = 1.
\end{aligned}
\end{align}
\end{lemma}
\begin{proof}
The proof is delegated to Appendix \ref{symmetry_properties}.
%\textbf{Proof of the equality in \eqref{property1}:}The equality in \eqref{property1} is a natural result of a symmetric network. When we swap the locations of the source and the destination, it should not create a difference in a symmetric network. Since the link capacities $\ell_{1,0} = \ell_{3,2}$ and $\ell_{2,0} = \ell_{3,1}$ are the same, the corresponding link activation times should also be the same. This proves the equality in \eqref{property1}.\\
%\textbf{Proof of the equality in \eqref{property2}:}
%Missing at this point!
\end{proof}

Based on Lemma~\ref{properties} and the definition of link capacities in~\eqref{modified_linkcap}, it follows that LP P3 is equivalent to the following LP.
\begin{align}\label{capacity_lp_cut_N_2_P4}
%\begin{align}
% \begin{array}{llll}
& \ \rm{P4:}\ \widebar{\msf{C}} = \max_{\lambda, \alpha} \alpha  \nonumber \\
&({\rm P4}a) \ \lambda_1 + \lambda_2 = 1,  \nonumber\\
&({\rm P4}b) \ \lambda_2 + \lambda_3 = 1, \nonumber\\
&({\rm P4}c) \ \lambda_1, \lambda_2, \lambda_3 \geq 0,\\
&({\rm P4}d) \ \alpha \leq \lambda_{2}\ell_2+\lambda_{1}\ell_{1} =\lambda_{2}\ell_2+(1-\lambda_2)\ell_1,  \nonumber \\
&({\rm P4}e) \ \alpha \leq \lambda_{1}\ell_{1}+\lambda_{2}\ell_{2} = (1-\lambda_2)\ell_1+\lambda_2\ell_2,  \nonumber \\
&({\rm P4}f) \ \alpha \leq \lambda_{3}\ell_{3}+2\lambda_{2}\ell_{2} = (1-\lambda_2)\ell_3+2\lambda_2\ell_2,  \nonumber\\
&({\rm P4}g) \ \alpha \leq 2\lambda_{1}\ell_{1} = 2(1-\lambda_2)\ell_1,  \nonumber
% \end{array}
%\end{align}
\end{align}
where $\ell_1$, $\ell_2$ and $\ell_3$ are defined as in~\eqref{modified_linkcap}, and $\lambda_3 = \lambda_{\ell_{2,1}}$. 

\medskip

\noindent{\bf [4. Proving Theorem~\ref{main_theorem} for different ranges of $\beta$]} The LP P4 will be the central starting point for the remainder of the proof. We will consider further sub-classes of our projected topologies depending on the value of $\beta \in (0,1/2]$. In particular, we consider the following two categories:
\begin{equation}\label{categories}
\begin{aligned}
% &{\rm\bf (Category\ 1)} :\quad  \widehat{\beta} \leq \beta \leq 1/2, \\
{\rm\bf (Category\ 1)} &:\quad  \frac{d}{\gamma^{\frac{1}{a}}} < \beta \leq 1/2,\\
{\rm\bf (Category\ 2)} &:\quad 0 < \beta \leq \frac{d}{\gamma^{\frac{1}{a}}}.\\
% &\text{where }\ \  \widehat{\beta} = -\tfrac{\gamma^{\frac{1}{a}}}{d}+1+\sqrt{\left(\tfrac{\gamma^{\frac{1}{a}}}{d}-1\right)^2-1+\frac{\gamma^{\frac{1}{a}}}{d}}.
\end{aligned}
\end{equation}

For each of these two categories, we will follow a different approach in order to show that there does not exist a feasible solution for LP P4 such that its objective function is greater than or equal to $\widebar{\msf{C}}^\star = 2\log\left(\gamma/d^a\right)$. 

The details of the proof for each category are described in the following subsections.

\subsection{Proof for Category 1}
To prove that Theorem~\ref{main_theorem} for Category 1 of $\beta$ values in~\eqref{categories}, it is sufficient to show that for $\beta > d/(\gamma^{1/a})$, the RHS of $({\rm P4}d)$ in~\eqref{capacity_lp_cut_N_2_P4} is smaller than $\widebar{\msf{C}}^\star$.

We prove this by contradiction. Let us assume that there does exist a value of $\beta > d/(\gamma^{1/a})$ and $\lambda_2 \in [0,1]$ such that the RHS of $({\rm P4}d)$ is greater than or equal to $\widebar{\msf{C}}^\star$, i.e., $\exists \beta$ such that
\begin{align}
&\widebar{\msf{C}}^\star \leq \text{RHS $({\rm P4}d)$ } = \lambda_{2}\ell_2+(1-\lambda_2)\ell_1 \nonumber \\
&= \lambda_{2}\log\left(\frac{\gamma}{(1-\beta)^ad^a}\right) + (1-\lambda_2)\log\left(\frac{\gamma}{\beta^ad^a}\right) \nonumber \\
&= \log\left(\frac{\gamma}{d^a}\right)-\lambda_2\log\left((1-\beta)^a\right)-(1-\lambda_2)\log\left(\beta^a\right),
\end{align}
With rearrangement of terms (recall the definition of $\widebar{\msf{C}}^\star$ in Corollary~\ref{optimum_cap}), and exploiting the fact that the logarithm function is monotonically increasing, the above constraint can be equivalently rewritten as
\begin{equation}\label{contradiction_eq1}
\begin{aligned}
 d \geq \underbrace{\gamma^{\frac{1}{a}}\beta(1-\beta)^{\lambda_2}\beta^{-\lambda_2}}_{g(\lambda_2)}.
\end{aligned}
\end{equation}
From basic calculus, it is not difficult to verify that the derivative of the function $g(\lambda_2)$ is non-negative, for all $\beta \in (0,1/2]$, and thus $g(\lambda_2)$ is monotonically increasing in $\lambda_2$. Additionally, since $\lambda_2$ is a non-negative variable in LP P4 (see the constraint $({\rm P4}c)$), then~\eqref{contradiction_eq1} would imply that
\begin{equation}
\begin{aligned}
d \geq g(0) = \gamma^{\frac{1}{a}}\beta, 
\end{aligned}
\end{equation}
which contradicts our assumption that $\beta > d/(\gamma^{1/a})$.

Thus, we conclude that for all values of $\beta > d/(\gamma^{1/a})$, there does not exist a feasible solution in LP P4, such that the RHS of $({\rm P4}d)$ is greater than or equal to $\widebar{\msf{C}}^\star$. Thus, we have that
\[
\widebar{\msf{C}} \leq \max_{\lambda_2} \text{RHS $({\rm P4}d)$} < \widebar{\msf{C}}^\star.
\]
The concludes the proof for Category 1.

\subsection{Proof for Category 2}
To prove Theorem~\ref{main_theorem} for Category 2\textbf{} topologies, we aim to show that for $0 < \beta \leq d / \gamma^{\frac{1}{a}}$, any feasible solution for P4 will have either the RHS of $({\rm P4}d)$ or the RHS of $({\rm P4}f)$ less than $\widebar{\msf{C}}^\star$. This statement is summarized by the following lemma which is proved in Appendix~\ref{appendix_third_range}.
\begin{lemma}\label{third_range}
If $\frac{\gamma}{d^a} > 3^a$ and $0< \beta \leq d / \gamma^{\frac{1}{a}}$, then $\forall \lambda_2 \in [0,1]$, we have that
\begin{equation}
\begin{aligned}
\widebar{\msf{C}}^\star > \min\left\{\underbrace{\lambda_{2}\ell_2+(1-\lambda_2)\ell_1}_{\text{RHS $({\rm P4}d)$}}, \underbrace{(1-\lambda_2)\ell_3+2\lambda_2\ell_2}_{\text{RHS (${\rm P4}f)$}}\right\}
\end{aligned}
\end{equation}
where $\widebar{\msf{C}}^\star$ is given in Corollary~\ref{optimum_cap} and $\ell_1$, $\ell_2$ and $\ell_3$ are as in~\eqref{modified_linkcap}.
\end{lemma}
The result for Category 2 is the direct consequence of Lemma~\ref{third_range} since the approximate capacity in P4 is equal to the minimum of RHS of the constraints $({\rm P4}d-{\rm P4}g)$.
This concludes the proof of Theorem \ref{main_theorem}.

\section{Conclusion}\label{sec:concl}
In this paper, we studied the optimization of relay placement in a 1-2-1 FD network with $2$ relay nodes in order to maximize the approximate capacity of the network. Through our analysis, we observed that the approximate capacity is maximized by placing the relay nodes as close as possible to the source and the destination, respectively. This serves as a first step in understanding optimal topology design in directional networks and can pave the way towards understanding and optimizing relay placement for more complex directional networks.
\bibliographystyle{IEEEtran}
\bibliography{bibliography}

\appendices
\section{Proof of Lemma~\ref{symmetric_network_lemma}}\label{asymm_symm}
Here, we prove that for an asymmetric network as in Fig.~\ref{projected_network}, we can find a symmetric network that gives at least the same capacity. The proof consists of three main steps. 

In the first step, we prove that the following result holds in asymmetric networks.
\begin{equation}\label{eq1_asymm}
\begin{aligned}
\lambda_{\ell_{3,1}}+\lambda_{\ell_{2,1}} = 1,\\
\lambda_{\ell_{2,0}}+\lambda_{\ell_{2,1}} = 1.
\end{aligned}
\end{equation}

As it was described in the first step of the proof outline, we define the distances between the nodes in the network using fractions $\beta_1$ and $\beta_2$ of the distance $d$ between the source and destination where $\beta,~\beta_1 \in [0,1]$. In the second step, we prove that when $\beta_1 < \beta_2$, $F_{3,1} < F_{2,0}$ in the optimal solution and through using this property, we show that there exists a symmetric network that gives at least the same capacity as the asymmetric network. 

In the third step, we prove that when $\beta_1 > \beta_2$, $F_{3,1} > F_{2,0}$ in the optimal solution and through leveraging this property, we again show that there exists a symmetric network that achieves at least the same capacity as the asymmetric network.

We start with the first step by considering an optimal solution where both $\lambda_{\ell_{3,1}}+\lambda_{\ell_{2,1}} < 1$ and $\lambda_{\ell_{2,0}}+\lambda_{\ell_{2,1}} < 1$. In this case, we can increase $\lambda_{\ell_{2,1}}$ until one of these summations becomes equal to $1$ while fixing all other variables from the optimal solution. This new set of variables satisfies the constraints in LP P1 and the same capacity is achieved. Therefore, there exists an optimal solution where one of these summations is equal to $1$. Now, consider the case where an optimal solution satisfies $\lambda_{\ell_{3,1}}+\lambda_{\ell_{2,1}} = 1$ and $\lambda_{\ell_{2,0}}+\lambda_{\ell_{2,1}} < 1$. We can represent these activation times as follows.
\begin{equation}\label{equation_appendix_a}
\begin{aligned}
\lambda_{\ell_{3,1}} = \lambda,\\
\lambda_{\ell_{2,1}} = 1-\lambda,\\
\lambda_{\ell_{2,0}} < \lambda.
\end{aligned}
\end{equation}

We note that in the optimal solution, $\lambda_{\ell_{1,2}} = 0$ since there exists an optimal solution where $F_{1,2} = 0$. Consider an optimal solution where $F_{1,2} \neq 0$. Due to the constraint $({\rm P1}b)$ in \eqref{capacity_lp} for the second relay,
\[
F_{1,2} = F_{2,1} + F_{2,0} - F_{3,2},
\]
While keeping other variables same, we can decrease $F_{1,2}$ to zero by reducing $F_{2,1} + F_{2,0}$ by the same amount without affecting $F_{3,2}$. In this case, we obtain another feasible solution that achieves the same capacity.

Now, for the activation times given in \eqref{equation_appendix_a}, there are two possible cases: $\lambda_{\ell_{30}} \geq \lambda-\lambda_{\ell_{2,0}}$ or $\lambda_{\ell_{30}} < \lambda-\lambda_{\ell_{2,0}}$.

In the first case, we can perform the following modifications on the activation times.
\begin{equation}\label{eq2_asymm}
\begin{aligned}
\lambda'_{\ell_{2,0}} = \lambda,\qquad \lambda'_{\ell_{3,0}} = \lambda_{\ell_{3,0}}-(\lambda-\lambda_{\ell_{2,0}}),\\
\lambda'_{\ell_{1,0}} = \lambda_{\ell_{1,0}}, \qquad \lambda'_{\ell_{3,2}} = \lambda_{\ell_{3,2}}+(\lambda-\lambda_{\ell{2,0}}),\\
\lambda'_{\ell_{3,1}} = \lambda, \qquad \lambda'_{\ell_{1,2}} = 0, \qquad \lambda'_{\ell_{2,1}} = 1-\lambda.
\end{aligned}
\end{equation}
This new set of activation times satisfy the constraints in LP P1. If $F_{3,0} \leq \lambda'_{\ell_{3,0}}\ell_{3,0}$, we can send the same flows and reach the same capacity. If $F_{3,0} > \lambda'_{\ell_{3,0}}\ell_{3,0}$, we can modify the flows as follows.
\begin{equation}
\begin{aligned}
    F'_{3,0} =  \lambda'_{\ell_{3,0}}\ell_{3,0},\\
    F'_{2,0} = F_{2,0}+F_{3,0}- \lambda'_{\ell_{3,0}}\ell_{3,0},\\
    F'_{3,2} = F_{3,2}+F_{3,0}- \lambda'_{\ell_{3,0}}\ell_{3,0}.
\end{aligned}
\end{equation}
The remaining flows stay same and the same capacity in LP P1 is achieved. We note that since the decrease in $\lambda_{\ell_{3,0}}$ is equal to the increase in $\lambda_{\ell_{2,0}}$ and $\lambda_{\ell_{3,2}}$, and the link capacity $\ell_{3,0}$ is the smallest link capacity in the network, these modified flows still satisfy the capacity constraints $({\rm P1}a)$ in LP P1.

In the second case, we have $\lambda_{\ell_{3,0}} < \lambda-\lambda_{\ell_{2,0}}$. Therefore, we can modify the activation times as follows.
\begin{equation}
\begin{aligned}
\lambda'_{\ell_{2,0}} = \lambda_{\ell_{2,0}}+\lambda_{\ell_{3,0}},\\
\lambda'_{\ell_{3,1}} = \lambda_{\ell_{2,0}}+\lambda_{\ell_{3,0}},\\
\lambda'_{\ell_{3,2}} = \lambda_{\ell_{3,2}}+\lambda_{\ell_{3,0}}+\lambda_{\ell_{3,1}}-\lambda'_{\ell_{3,1}},\\ \lambda'_{\ell_{1,0}} = \lambda_{\ell_{1,0}} \qquad
\lambda'_{\ell_{3,0}} = 0,\\ \lambda'_{\ell_{1,2}} = 0 \qquad \lambda'_{\ell_{2,1}} = 1-\lambda'_{\ell_{3,1}}.
\end{aligned}
\end{equation}
If $F_{3,0} \leq \lambda'_{\ell_{3,0}}\ell_{3,0}$ and $F_{3,1} \leq \lambda'_{\ell_{3,1}}\ell_{3,1}$, we can use the same flow variables and achieve the same capacity. Otherwise, the flows through the paths $p_1$ and/or $p_2$ need to decrease, respectively where $p_1$ is the direct path between the source and destination and $p_2$ is the path going from the source to the destination while passing through relay $1$. However, we can compensate the decrease in paths $p_1$ and $p_2$ by sending additional flows through the paths $p_3$ and $p_4$, respectively where $p_3$ is the path that goes from the source to the destination while passing through relay $2$ and $p_4$ is the path that goes from the source to the destination while passing through both relays. For example, if we have $F_{3,0} > \lambda'_{\ell_{3,0}}\ell_{3,0}$ and $F_{3,1} > \lambda'_{\ell_{3,1}}\ell_{3,1}$, we can modify the flows as follows.
\begin{equation}
\begin{aligned}
F'_{3,0} = 0, \qquad F'_{3,1} = \lambda'_{\ell_{3,1}}\ell_{3,1},\\
F'_{1,0} = F_{1,0},\qquad F'_{2,0} = F_{2,0}+F_{3,0},\\
F'_{1,2} = 0, \qquad F'_{2,1} = F_{2,1}+F_{3,1}-F'_{3,1}\\
F'_{3,2} = F_{3,2}+F_{3,0}+F_{3,1}-F'_{3,1}
\end{aligned}
\end{equation}
This new set of variables satisfies the constraints in LP P1 and achieve the same capacity. Therefore, in one of the optimal solutions, the equalities in \eqref{eq1_asymm} holds. We can perform similar steps for the case $\lambda_{\ell_{3,1}}+\lambda_{\ell_{2,1}} < 1$ and $\lambda_{\ell_{2,0}}+\lambda_{\ell_{2,1}} = 1$, and arrive the same conclusion. Therefore, in asymmetric networks, the following holds. 
\begin{equation}\label{first_step_asymm}
\begin{aligned}
\lambda_{\ell_{3,1}} = \lambda_{\ell_{2,0}}.
\end{aligned}
\end{equation}
This completes the proof of the first step.

In the second step, we leverage the property in \eqref{first_step_asymm} to show that $F_{3,1} < F_{2,0}$ in the optimal solution when $\beta_1 < \beta_2$ in an asymmetric network. We represent the activation times as follows.
\begin{equation*}
\begin{aligned}
\lambda_{\ell_{3,1}} = \lambda_{\ell_{2,0}} = \lambda,\\
\lambda_{\ell_{2,1}} = 1-\lambda,\\
\lambda_{\ell_{1,2}} = 0.
\end{aligned}
\end{equation*}
It is not difficult to see that in the optimal solution, $\lambda_{\ell_{3,0}}+\lambda_{\ell_{3,1}}+\lambda_{\ell_{3,2}} = 1$ and $\lambda_{\ell_{3,0}}+\lambda_{\ell_{2,0}}+\lambda_{\ell_{1,0}} = 1$. Consider there exists an optimal solution such that $\lambda_{\ell_{3,0}}+\lambda_{\ell_{3,1}}+\lambda_{\ell_{3,2}} < 1$ and $\lambda_{\ell_{3,0}}+\lambda_{\ell_{2,0}}+\lambda_{\ell_{1,0}} < 1$. While fixing other variables, if we increase $\lambda_{\ell_{3,2}}$ and $\lambda_{\ell_{1,0}}$ until the summations become equal to $1$, this will give another feasible solution that achieves the same capacity. As a result, there exists an optimal solution where $\lambda_{\ell_{3,0}}+\lambda_{\ell_{3,1}}+\lambda_{\ell_{3,2}} = 1$ and $\lambda_{\ell_{3,0}}+\lambda_{\ell_{2,0}}+\lambda_{\ell_{1,0}} = 1$. Therefore, 
\begin{equation*}
\lambda_{\ell_{3,2}} = \lambda_{\ell_{1,0}} = 1-\lambda-\lambda_{\ell_{3,0}}.
\end{equation*}
In order to show that $F_{3,1} < F_{2,0}$ in the optimal solution, we will show that $F_{2,0}$ satisfies the capacity constraint $({\rm P1}a)$ in \eqref{capacity_lp} with equality in the optimal solution, i.e., $F_{2,0} = \lambda\ell_{2,0}$. Since $\ell_{2,0} > \ell_{3,1}$ and $F_{3,1} \leq \lambda\ell_{3,1}$, this guarantees that $F_{3,1} < F_{2,0}$ in the optimal solution.

We start with considering an optimal solution where $F_{2,0} < \lambda\ell_{2,0}$. In order to prove the argument, we will evaluate two cases. In the first case, $F_{3,2}+b_1 \leq (1-\lambda-\lambda_{\ell_{3,0}})\ell_{3,2}$ and in the second case, $F_{3,2}+b_1 > (1-\lambda-\lambda_{\ell_{3,0}})\ell_{3,2}$ where $b_1 = \lambda\ell_{2,0}-F_{2,0}$.

In the first case, we can use contradiction to prove the argument. Consider $F_{2,0} < \lambda\ell_{2,0}$ in the optimal solution. We can modify the flow variables as follows.
\begin{equation}
\begin{aligned}
F'_{2,0} = \lambda\ell_{2,0},\\
F'_{3,2} = F_{3,2}+b_1.
\end{aligned}
\end{equation}
where $b_1 = \lambda\ell_{2,0}-F_{2,0}$ and the remaining variables stay same. 

With this new set of variables, we have another feasible solution that achieves a higher capacity than the optimal solution since $F_{3,2}$ increases, $F_{3,1}$ and $F_{3,0}$ stay same. This creates a contradiction, therefore, we argue that an optimal solution satisfies $F_{2,0} = \lambda\ell_{2,0}$ if $F'_{3,2}$ satisfies the capacity constraint $({\rm P1}a)$ in \eqref{capacity_lp}, i.e.,
\begin{equation*}
F_{3,2}+b_1 \leq (1-\lambda-\lambda_{\ell_{3,0}})\ell_{3,2}.
\end{equation*}

In the second case, we have $F_{3,2}+b_1 > (1-\lambda-\lambda_{\ell_{3,0}})\ell_{3,2}$. We can decrease $\lambda$ to $\lambda'_1$ where $\lambda'_1$ satisfies 
\begin{equation}\label{eq1_second_step_asymm}
\begin{aligned}
F_{3,2}+b'_1 = (1-\lambda'_1-\lambda_{\ell_{3,0}})\ell_{3,2},\\
\implies \lambda'_1 = \frac{\ell_{3,2}-\lambda_{\ell_{3,0}}\ell_{3,2}+F_{2,0}-F_{3,2}}{\ell_{2,0}+\ell_{3,2}}.
\end{aligned}
\end{equation}
where $b'_1 = \lambda'_1\ell_{2,0}-F_{2,0} > 0$. We should note that $\lambda'_1$ in \eqref{eq1_second_step_asymm} is a feasible activation time since  $F_{3,2} \leq (1-\lambda-\lambda_{\ell_{3,0}})\ell_{3,2}$. Now, we need to consider two cases: $F_{3,1} \leq \lambda'_1\ell_{3,1}$ and $F_{3,1} > \lambda'_1\ell_{3,1}$.

In the first case, we can use contradiction to prove our argument. First, we modify the flow variables as follows.
\begin{equation}
\begin{aligned}
F'_{2,0} = \lambda'_1\ell_{2,0},\\
F'_{3,2} = F_{3,2}+b'_1.
\end{aligned}
\end{equation}
When the remaining flow variables stay same, this new set of variables is a feasible solution and achieves a higher capacity than the optimal solution since $F_{3,2}$ increases, $F_{3,1}$ and $F_{3,0}$ stay same. Therefore, this creates a contradiction and we argue that $F_{2,0}$ satisfies capacity constraint $({\rm P1}a)$ in LP P1 in an optimal solution if $F_{3,1} \leq \lambda'_1\ell_{3,1}$.

In the second case ($F_{3,1} > \lambda'_1\ell_{3,1}$), we modify the flow variables as follows.
\begin{equation}
\begin{aligned}
    F'_{3,1} = \lambda'_1\ell_{3,1},\\
    F'_{2,0} = \lambda'_1\ell_{2,0},\\
    F'_{1,0} = F_{2,1}+F'_{3,1},\\
    F'_{3,2} = F_{3,2}+b'_1.
\end{aligned}
\end{equation}
The remaining flow variables stay same. This new set of variables still satisfies the constraints in LP P1 and achieves a higher capacity since the increase in $F_{3,2}$ is higher than the decrease in $F_{3,1}$.
\begin{equation}
\begin{aligned}
F_{3,1}-F'_{3,1} \leq (\lambda-\lambda'_1)\ell_{3,1}\\
F'_{3,2}-F_{3,2} \geq (\lambda-\lambda'_1)\ell_{3,2}
\end{aligned}
\end{equation}
Since $\ell_{3,2} > \ell_{3,1}$, this new set of variables achieve a higher capacity and again, this creates a contradiction. Therefore, we argue that in an optimal solution, $F_{2,0}$ satisfies the capacity constraint with equality. As a result, we claim that
\begin{equation}\label{eq2_second_step_asymm}
F_{3,1} < F_{2,0} \qquad \text{if}~\beta_1 < \beta_2
\end{equation}
Through using the property in \eqref{eq2_second_step_asymm}, we can prove that there exists a symmetric network that achieves at least the same capacity as the asymmetric network. Consider we have an asymmetric network as in Fig.~\ref{projected_network} and $\beta_1 < \beta_2$. Now, we bring relay $1$ closer to relay $2$ such that $\beta'_1 = \beta_2$. We next show that with this symmetric network, we can achieve the same capacity. Towards this end, we first write the relationships between variables in an optimal solution for the asymmetric network based on constraints in LP P1, proofs given above and the property in \eqref{eq2_second_step_asymm}.
\begin{equation}
\begin{aligned}
\lambda_{\ell_{3,1}} = \lambda_{\ell_{2,0}} = \lambda,\\
\lambda_{\ell_{1,0}} = \lambda_{\ell_{3,2}} = 1-\lambda-\lambda_{\ell_{3,0}},\\
\lambda_{\ell_{2,1}} = 1-\lambda,\\
\lambda_{\ell_{1,2}} = 0.
\end{aligned}
\end{equation}
\begin{equation}\label{eq3_second_step_asymm}
\begin{aligned}
F_{3,1} < F_{2,0} = \lambda\ell_{2,0},\\
F_{1,0} = F_{3,1}+F_{2,1} < F_{3,2}, \\
F_{3,2}= F_{2,0}+F_{2,1} \leq (1-\lambda-\lambda_{\ell_{3,0}})\ell_{3,2}.
\end{aligned}
\end{equation}
Now, we bring relay $1$ closer to relay $2$ such that $\beta'_1 = \beta_2$. In this case, the distances between nodes are modified as follows.
\begin{equation*}
\begin{aligned}
d'_{3,1} = d_{2,0},\\
d'_{1,0} = d_{3,2},\\
d'_{2,1} = d-2\beta_2d.
\end{aligned}
\end{equation*}
Through the path loss model in \eqref{channel_model}, 
\begin{equation*}
\begin{aligned}
\ell'_{3,1} = \ell_{2,0},\\
\ell'_{1,0} = \ell_{3,2},\\
\ell'_{2,1} > \ell_{2,1}.
\end{aligned}
\end{equation*}
Now, we use the same activation times $\lambda_{\ell_{j,i}}$ as in the asymmetric network. Thus, we have the following capacity constraints on $F_{3,1}$ and $F_{1,0}$ in the symmetric network.
\begin{equation}
\begin{aligned}
F_{3,1} \leq \lambda\ell_{2,0}, \\
F_{1,0} \leq (1-\lambda-\lambda_{\ell_{3,0}})\ell_{3,2},\\
F_{2,1} \leq (1-\lambda)\ell'_{2,1}.
\end{aligned}
\end{equation}
These capacity constraints are satisfied due to the relationship given in \eqref{eq3_second_step_asymm}. Therefore, we can use the same activation times and flow variables in this symmetric network, and this gives a feasible solution that achieves the same capacity as the asymmetric network. This concludes the proof for the case where $\beta_1 < \beta_2$.

In the third step, we show that $F_{3,1} > F_{2,0}$ in the optimal solution of an asymmetric network when $\beta_1 > \beta_2$. Then, by leveraging this property, we show that there exists a symmetric network that achieves at least the same capacity as the asymmetric network. 

We first represent the optimal activation times in an asymmetric network as follows based on our result in \eqref{first_step_asymm}.
\begin{equation}
\begin{aligned}
\lambda_{\ell_{3,1}} = \lambda_{\ell_{2,0}} = \lambda,\\
\lambda_{\ell_{1,0}} = \lambda_{\ell_{3,2}} = 1-\lambda-\lambda_{\ell_{3,0}},\\
\lambda_{\ell_{2,1}} = 1-\lambda,\\
\lambda_{\ell_{1,2}} = 0.
\end{aligned}
\end{equation}
In order to show that $F_{3,1} > F_{2,0}$ in the optimal solution, we will show that $F_{3,1}$ satisfies the capacity constraint $({\rm P1}a)$ in LP P1 with equality in the optimal solution, i.e., $F_{3,1} = \lambda\ell_{3,1}$. Since $\ell_{2,0} < \ell_{3,1}$ and $F_{2,0} \leq \lambda\ell_{2,0}$, this guarantees that $F_{2,0} < F_{3,1}$ in the optimal solution. 

Consider there is an optimal solution where $F_{3,1} < \lambda\ell_{3,1}$. Through using contradiction, we show that $F_{3,1} = \lambda\ell_{3,1}$. Now, we consider two cases: $F_{1,0}+b_2 \leq (1-\lambda-\lambda_{\ell_{3,0}})\ell_{1,0}$ and
$F_{1,0}+b_2 > (1-\lambda-\lambda_{\ell_{3,0}})\ell_{1,0}$ where $b_2 = \lambda\ell_{3,1}-F_{3,1}$.

In the first case, we modify the flow variables as follows.
\begin{equation}
\begin{aligned}
F'_{3,1} = \lambda\ell_{3,1},\\
F'_{1,0} = F_{1,0}+b_2.
\end{aligned}
\end{equation}
The remaining variables stay same. This new set of variables is still feasible in LP P1 and it achieves a higher capacity than the optimal solution due to increase in $F_{3,1}$ and $F_{1,0}$. Since this creates a contradiction, we argue that an optimal solution satisfies $F_{3,1} = \lambda\ell_{3,1}$ if $F'_{1,0}$ satisfies the capacity constraint, i.e., $F_{1,0}+b_2 \leq (1-\lambda-\lambda_{\ell_{3,0}})\ell_{1,0}$.

In the second case, we have $F_{1,0}+b_2 > (1-\lambda-\lambda_{\ell_{3,0}})\ell_{1,0}$. We can decrease $\lambda$ to $\lambda'_2$ where $\lambda'_2$ satisfies 
\begin{equation}\label{eq1_third_step_asymm}
\begin{aligned}
F_{1,0}+b'_2 = (1-\lambda'_2-\lambda_{\ell_{3,0}})\ell_{1,0},\\
\implies \lambda'_2 = \frac{\ell_{1,0}-\lambda_{\ell_{3,0}}\ell_{1,0}+F_{3,1}-F_{1,0}}{\ell_{3,1}+\ell_{1,0}}.
\end{aligned}
\end{equation}
where $b'_2 = \lambda'_2\ell_{3,1}-F_{3,1} > 0$.

We should note that $\lambda'_2$ in \eqref{eq1_third_step_asymm} is a feasible activation time since  $F_{1,0} \leq (1-\lambda-\lambda_{\ell_{3,0}})\ell_{1,0}$. Now, we consider two cases: $F_{2,0} \leq\lambda'_2\ell_{2,0}$ and $F_{2,0} > \lambda'_2\ell_{2,0}$.

In the first case, we modify the following flow variables.
\begin{equation}
\begin{aligned}
F'_{3,1} = \lambda'_2\ell_{3,1},\\
F'_{1,0} = F_{1,0}+b'_2.
\end{aligned}
\end{equation}
When the remaining flow variables stay same, this new set of variables gives a feasible solution for LP P1 and achieves a higher capacity than the optimal solution due to increase in $F_{3,1}$. Therefore, this causes a contradiction again and shows that $F_{3,1}$ satisfies the capacity constraint $({\rm P1}a)$ in \eqref{capacity_lp} in an optimal solution.

In the second case, we modify the following flow variables.
\begin{equation}
\begin{aligned}
F'_{2,0} = \lambda'_2\ell_{2,0},\\
F'_{3,1} = \lambda'_2\ell_{3,1},\\
F'_{3,2} = F'_{2,0}+F_{2,1},\\
F'_{1,0} = F'_{3,1}+F_{2,1}.
\end{aligned}
\end{equation}
The remaining variables stay same. This new set of variables gives a feasible solution and achieves a higher capacity since the increase in $F_{1,0}$ is larger than decrease in $F_{2,0}$.
\begin{equation}
\begin{aligned}
F_{2,0}-F'_{2,0} \leq (\lambda-\lambda'_2)\ell_{2,0},\\
F'_{1,0}-F_{1,0} \geq (\lambda-\lambda'_2)\ell_{1,0}.
\end{aligned}
\end{equation}
Since $\ell_{1,0} > \ell_{2,0}$, the capacity $F_{3,0}+F_{3,1}+F_{3,2} = F_{1,0}+F_{2,0}+F_{3,0}$ increases. This creates a contradiction because this new set of variables achieve a higher capacity than the optimal solution for the asymmetric network. As a result, we claim that the following holds in an asymmetric network when $\beta_1 > \beta_2$.
\begin{equation}\label{eq2_third_step_asymm}
\begin{aligned}
F_{2,0} < F_{3,1}.
\end{aligned}
\end{equation}
Through using the property in \eqref{eq2_third_step_asymm}, we can prove that there exists a symmetric network that gives at least the same capacity as the asymmetric network. Consider we have an asymmetric network as in Fig.~\ref{projected_network} and $\beta_1 > \beta_2$. Now, we bring relay $2$ closer to relay $1$ such that $\beta_1 = \beta'_2$. We next show that with this symmetric network, we can achieve the same capacity. Towards this end, we first write the relationships between variables in an optimal solution for the asymmetric network based on constraints in LP P1, proofs given above and the property in \eqref{eq2_third_step_asymm}.
\begin{equation}
\begin{aligned}
\lambda_{\ell_{3,1}} = \lambda_{\ell_{2,0}} = \lambda,\\
\lambda_{\ell_{1,0}} = \lambda_{\ell_{3,2}} = 1-\lambda-\lambda_{\ell_{3,0}},\\
\lambda_{\ell_{2,1}} = 1-\lambda,\\
\lambda_{\ell_{1,2}} = 0.
\end{aligned}
\end{equation}
\begin{equation}\label{eq3_third_step_asymm}
\begin{aligned}
F_{2,0} < F_{3,1} = \lambda\ell_{3,1},\\
F_{3,2} = F_{2,0}+F_{2,1} < F_{1,0}, \\
F_{1,0} = F_{3,1}+F_{2,1} \leq (1-\lambda-\lambda_{\ell_{3,0}})\ell_{1,0}.
\end{aligned}
\end{equation}
Now, we bring relay $2$ closer to relay $1$ such that $\beta_1 = \beta'_2$. In this case, the distances between nodes are modified as follows.
\begin{equation*}
\begin{aligned}
d'_{2,0} = d_{3,1},\\
d'_{3,2} = d_{1,0},\\
d'_{2,1} = d-2\beta_1d.
\end{aligned}
\end{equation*}
Through the path loss model in \eqref{channel_model}, 
\begin{equation*}
\begin{aligned}
\ell'_{2,0} = \ell_{3,1},\\
\ell'_{3,2} = \ell_{1,0},\\
\ell'_{2,1} > \ell_{2,1}.
\end{aligned}
\end{equation*}
Now, we use the same activation times $\lambda_{\ell_{j,i}}$ as in the optimal solution of the asymmetric network. Thus, we have the following capacity constraints on $F_{2,0}$ and $F_{3,2}$ in the symmetric network.
\begin{equation}
\begin{aligned}
F_{2,0} \leq \lambda\ell_{3,1}\\
F_{3,2} \leq (1-\lambda-\lambda_{\ell_{3,0}})\ell_{1,0}
\end{aligned}
\end{equation}
These capacity constraints are satisfied due to the relationship given in \eqref{eq3_third_step_asymm}. Therefore, we can use the same activation times and flow variables in this symmetric network, and this gives a feasible solution that achieves the same capacity as the asymmetric network. This concludes the proof for the case where $\beta_1 > \beta_2$.

\section{Proof of Lemma~\ref{properties}}\label{symmetry_properties}
\subsection{Proof of the properties in \eqref{property1}} 
Here, we prove this property for the symmetric projected networks, considered in Lemma~\ref{properties}. From the symmetry property of the network, it is not difficult to see that for any optimal solution $\{\lambda_{\ell_{j,i}}\}$ of the LP P3 in~\eqref{capacity_lp_cut_N_2}, we can construct another optimal solution $\{\lambda'_{\ell_{j,i}}\}$ by setting
\begin{align*}
    \lambda'_{\ell_{1,0}} &= \lambda_{\ell_{3,2}}, \qquad  \lambda'_{\ell_{3,2}} = \lambda_{\ell_{1,0}},\\  
    \lambda'_{\ell_{2,0}} &= \lambda_{\ell_{3,1}}, \qquad  \lambda'_{\ell_{3,1}} = \lambda_{\ell_{2,0}},\\
    \lambda'_{\ell_{2,1}} &= \lambda_{\ell_{2,1}}, \qquad \lambda'_{\ell_{1,2}} = \lambda_{\ell_{1,2}},\\
    \lambda'_{\ell_{3,0}} &= \lambda_{\ell_{3,0}}.
\end{align*}
Now since the LP P4 has an affine objective function, then the midpoint between $\lambda'$ and $\lambda$ is also optimal. Thus, there exists an optimal solution that satisfies the property in~\eqref{property1}.
% The equality in \eqref{property1} is a natural result of a symmetric network. When we swap the locations of the source and the destination, it should not create a difference in a symmetric network. Since the link capacities $\ell_{1,0} = \ell_{3,2}$ and $\ell_{2,0} = \ell_{3,1}$ are the same, the corresponding link activation times should also be the same. This proves the equality in \eqref{property1}.\\

\medskip

\subsection{Proof of the properties in \eqref{property2}}
We first note that the LP P3 and the LP P1 share the same $\lambda_{j,i}$ variables for the case $N=2$. 
Now, we focus on the case where $F_{1,2} \neq 0$ in the LP P3, since otherwise $\lambda_{1,2}$ can be pushed to zero without consequence on the objective function is trivial. From the condition $({\rm P1}b)$ in~\eqref{capacity_lp}, we have the following for any feasible point
\[
F_{1,2} = F_{2,1} + F_{2,0} - F_{3,2},
\]
thus, we can reduce the flow $F_{1,2}$ to zero and equally reduce the flow of the sum $F_{2,1} + F_{2,0}$ without affecting $F_{3,2}$. Thus, we arrive at another feasible point with the same objective value while having $F_{1,2} = 0$. Reducing $\lambda_{1,2}$ subsequently still maintains a feasible point and doesn't affect the objective function. 
% This concludes the proof of the property \eqref{property2}.
% assume that note that  is the result of the relay locations in Figure \ref{projected_network}. Since relay $2$ is closer to the destination, it would not be an optimal solution if it sent its flow to the relay $1$ instead of the destination.\\

For the proof of $\lambda_2+\lambda_{\ell_{2,1}} = 1$, we assume that we have an optimal solution for P3 in~\eqref{capacity_lp_cut_N_2} satisfying the property~\eqref{property1} and that $\lambda_{\ell_{1,2}} = 0$ as proved above.  From~\eqref{property1}, we have that
\begin{align*}
    \lambda_{1} = \lambda_{\ell_{1,0}} = \lambda_{\ell_{3,2}}, \\
    \lambda_{2} = \lambda_{\ell_{3,1}} = \lambda_{\ell_{2,0}}.
\end{align*}
Now given the aforementioned optimal solution for P3, the only upper bounding constraint we have on $\lambda_{\ell_{2,1}}$, is given by
\[
    \lambda_{\ell_{2,1}} + \lambda_2 \leq 1.
\]
While fixing all other variables from the optimal solution, it is not difficult to see that increasing $\lambda_{\ell_{2,1}}$ (if needed) so that the inequality is satisfied with equality leads to a feasible point with the same objective function value as the optimal. Thus there exists an optimal solution satisfying that $\lambda_{\ell_{2,1}} + \lambda_2 = 1$.
% We first assume that there is an optimal solution $\lambda = (\lambda_1^{(1)}, \lambda_2^{(1)},\lambda_{\ell_{2,1}}^{(1)},\lambda_{\ell_{3,0}}^{(1)})$ such that $\lambda_2+\lambda_{\ell_{2,1}} < 1$. We denote the flow variables for this optimal solution as $F_{j,i}^{(1)}$. In the second step, we only increase $\lambda_{\ell_{2,1}}$ such that $\lambda_2+\lambda_{\ell_{2,1}} = 1$ and all other activation times stay same. In this case, it is not difficult to see that the activation time constraints in \eqref{capacity_lp} are still satisfied since the only constraint that changes after this modification is $\lambda_2+\lambda_{\ell_{2,1}} \leq 1$ and it is satisfied with equality. If we use the same flows $F_{j,i}^{(1)}$ in this step, the capacity constraints in \eqref{capacity_lp} are still satisfied since all $\lambda_{\ell_{j,i}}\ell_{j,i}$ terms stay same or increase. Since we use the same flow variables, we achieve the same capacity as the optimal solution by using this new set of activation times. Therefore, we can claim that at least one of the optimal solutions satisfies $\lambda_2+\lambda_{\ell_{2,1}} = 1$. This proves the equality in \eqref{property2}.\\

\medskip

\subsection{Proof of the properties in \eqref{property3}} 
In order to prove that the optimal solution satisfies $\lambda_{\ell_{3,0}} = 0$ and $\lambda_1+\lambda_2 = 1$, we consider two cases: $\ell_{2,1} \geq \ell_1$ and $\ell_{2,1} < \ell_1$ where $\ell_1$ is defined in~\eqref{modified_linkcap}.

In the first case, we consider that there is an optimal solution $\left\{\lambda^{(1)}_{\ell_{j,i}}\right\}$ and $\left\{F^{(1)}_{j,i}\right\}$ for LP P1 in~\eqref{capacity_lp} such that
\begin{equation*}
\begin{aligned}
\lambda_1^{(1)}+\lambda_2^{(1)}+\lambda_{\ell_{3,0}}^{(1)} = 1,\\
\lambda_2^{(1)}+\lambda_{\ell_{2,1}}^{(1)} = 1.
%\lambda_{\ell_{3,0}}^{(1)} \geq 0.
\end{aligned}
\end{equation*}
where $\lambda_1$ and $\lambda_2$ are defined in~\eqref{property1}, and we use the properties in~\eqref{property1} and \eqref{property2}.

We note that there exists an optimal solution that satisfies $\lambda_1^{(1)}+\lambda_2^{(1)}+\lambda_{\ell_{3,0}}^{(1)} \leq 1$ constraint in LP P1 with equality. Otherwise, we can increase $\lambda^{(1)}_{\ell_{3,0}}$ until the constraint is satisfied with equality while fixing other variables. In this case, it is clear that all constraints in LP P1 are still satisfied by this new solution and the same capacity is achieved. 

Now, if we assume that there is an optimal solution where $\lambda^{(1)}_{\ell_{3,0}} > 0$, we can modify the activation times as follows.
\begin{equation}
\begin{aligned}
\lambda_{\ell_{3,0}}^{(2)} = 0,\\
\lambda_1^{(2)} = \lambda_1^{(1)}+\lambda_{\ell_{3,0}}^{(1)}.
\end{aligned}
\end{equation}
The remaining activation times stay same. Moreover, we modify the flow variables $F_{j,i}^{(1)}$ as follows.
\begin{equation}
\begin{aligned}
F_{1,0}^{(2)} = F_{1,0}^{(1)}+F_{3,0}^{(1)},\\
F_{2,1}^{(2)} = F_{2,1}^{(1)}+F_{3,0}^{(1)},\\
F_{3,2}^{(2)} = F_{3,2}^{(1)}+F_{3,0}^{(1)},\\
F_{3,0}^{(2)} = 0.\\
\end{aligned}
\end{equation}
The remaining flow variables stay same.

Through this modification, we reduce the flow of the direct link from the source to the destination to zero and send the flow $F^{(1)}_{3,0}$ through the path $p$ where $p$ is the path from the source to the destination and it passes through both of the relays. It is not difficult to see that the modified activation times still satisfy the constraints in \eqref{capacity_lp} and the flow conservation constraints $({\rm P1}b)$ are satisfied as well. We next show that the following capacity constraints in $({\rm P1}a)$ for the modified flow variables are still satisfied when $\ell_{2,1} \geq \ell_1$.
\begin{equation}
\begin{aligned}
    F_{1,0}^{(2)} = F_{1,0}^{(1)}+F_{3,0}^{(1)} \leq \lambda_1^{(1)}\ell_1+\lambda_{\ell_{3,0}}^{(1)}\ell_{3,0} \leq \lambda_1^{(2)}\ell_1,\\
    F_{3,2}^{(2)} = F_{3,2}^{(1)}+F_{3,0}^{(1)} \leq \lambda_1^{(1)}\ell_1+\lambda_{\ell_{3,0}}^{(1)}\ell_{3,0} \leq \lambda_1^{(2)}\ell_1,\\
    F_{2,1}^{(2)} = F_{2,1}^{(1)}+F_{3,0}^{(1)} \leq F_{1,0}^{(1)}+F_{3,0}^{(1)} \leq \lambda_1^{(2)}\ell_1 \leq \lambda_{\ell_{2,1}}^{(2)}\ell_3.
\end{aligned}
\end{equation}
where $\ell_3 = \ell_{2,1}$.

In the first two constraints, we use the fact that $\ell_{3,0}$ is the smallest link capacity in the network. In the last one, we first use  $F_{2,1}^{(1)} \leq F_{1,0}^{(1)}$ due to the constraint in $({\rm P1}b)$ and then, we use the fact that $\lambda_{\ell_{2,1}}^{(1)} = \lambda_{\ell_{2,1}}^{(2)} = \lambda_1^{(2)}$. Since this new set of variables provides a feasible solution that achieves the same capacity as the optimal solution, we claim that at least one of the optimal solutions will satisfy $\lambda_{\ell_{3,0}} = 0$ and $\lambda_1+\lambda_2 = 1$ if $\ell_3 \geq \ell_1$. 

In the second case,  we assume that $\ell_3 < \ell_1 $ and we use the following linear program proposed in \cite{ezzeldin}.   
\begin{align}
\label{capacity_paths}
\begin{array}{llll}
&\ \rm{P5:}\ \widebar{\msf{C}} = {\rm max}  \displaystyle\sum_{p \in \mathcal{P}} x_p \mathsf{C}_p   & & \\
&      ({\rm P5}a) \ x_p \geq 0 & \forall p \!\in\! \mathcal{P}, &  \\
&    ({\rm P5}b) \ \displaystyle\sum_{p \in \mathcal{P}_i}  x_p f^p_{p\pnext(i),i} \!\leq\! 1 & \forall i \! \in \! [0\!:\!N], & \\
&   ({\rm P5}c) \ \displaystyle\sum_{p \in \mathcal{P}_i} x_p f^p_{i,p\pprev(i)} \!\leq\! 1 & \forall i \! \in \! [1\!:\!N\!+\!1],  &
\end{array}
\end{align}
where $\mathcal{P}$ is the collection of all paths going from the source to the destination, $C_p$ is the capacity of path $p$, $\mathcal{P}_i \subseteq \mathcal{P}$ is the collection of paths passing through node $i$ where $i \in [0:N+1]$, $p\pnext(i)$ (respectively, $p\pprev(i)$) is the node that follows (respectively, precedes) node $i$ in path $p$, the variable $x_p$ is the fraction of time path $p$ is used and $f^p_{j,i}$ is the optimal activation time for the link of capacity $\ell_{j,i}$ when path $p$ is operated, i.e.,
\begin{equation*}
    f^p_{j,i} = \frac{C_p}{\ell_{j,i}}.
\end{equation*}

In \cite{ezzeldin}, it is proved that this program is an equivalent program to the one in \eqref{capacity_lp}. We can particularly write this program for our symmetric and projected network under the condition $\ell_3 < \ell_1 $. There are four paths $p_1$, $p_2$, $p_3$ and $p_4$ in our network where $p_1$ is the direct path between the source and the destination, $p_2$ (respectively, $p_3$) is the path going from the source to the destination while passing through relay $1$ (respectively, relay $2$) and path $p_4$ is the path going from the source to the destination while passing through both relays.
\begin{align}
\label{primal_lp}
\begin{array}{llll}
%\displaystyle
&\ \rm{P6:}\ \widebar{\msf{C}} ={\rm max} \displaystyle  \ x_1\ell_{3,0}+2x_2\ell_2+x_4\ell_3 &&\\
&({\rm P6}a) \ x_1+x_2\left(\frac{\ell_2}{\ell_1}+1\right)+x_4\frac{\ell_3}{\ell_1} \leq 1&\\ & ({\rm P6}b) \ x_2+x_4 \leq 1&\\ &({\rm P6}c)\ x_1, x_2, x_4 \geq 0&
\end{array}
\end{align}
where $\ell_1$, $\ell_2$ and $\ell_3$ are given in \eqref{modified_linkcap}.\\
Here, we take $x_2$ = $x_3$ due to the symmetric nature of our network. Our aim is to show that the optimal solution of this linear program satisfies $x_1 = 0$. Hence, the optimal solution does not send information through the direct path between the source and the destination. Since the linear program in \eqref{primal_lp} is equivalent to the program in \eqref{capacity_lp}, the corresponding $\lambda_{\ell_{3,0}}$ is equal to zero in this case.

Towards this end, we first find the dual of the primal program given in \eqref{primal_lp}. Then, we show that the feasible point $x = (x_1,x_2,x_4) = (0,\frac{\ell_1-\ell_3}{\ell_1+\ell_2-\ell_3},\frac{\ell_2}{\ell_1+\ell_2-\ell_3})$ satisfies the KKT conditions, therefore, it is an optimal solution to the program in \eqref{primal_lp}.

It is not difficult to derive the following dual program of the linear program in \eqref{primal_lp}. 
\begin{align}
\label{dual_lp}
\begin{array}{llll}
\displaystyle
&\ \rm{D1:}\ {\rm min} \displaystyle \ v_1+v_2\\
&({\rm D1}a) \ -\ell_{3,0}+v_1-s_1 = 0\\ & ({\rm D1}b) \ -2\ell_2+\frac{v_1\ell_2}{\ell_1}+v_1+v_2-s_2 = 0\\ &({\rm D1}c)\ -\ell_3+\frac{v_1\ell_3}{\ell_1}+v_2-s_3 = 0\\&({\rm D1}d) \ v_1,v_2,s_1,s_2,s_3 \geq 0
\end{array}
\end{align}
where $v_1, v_2, s_1,s_2,s_3$ are dual variables.\\ Then, we find the following complementary slackness conditions.
\begin{equation}
\begin{aligned}
    v_1\left(x_1+x_2\left(\frac{\ell_2}{\ell_1}+1\right)+x_4\frac{\ell_3}{\ell_1}-1\right) = 0\\
    v_2(x_2+x_4-1) = 0\\
    s_1x_1 = 0,\ s_2x_2 = 0,\ s_3x_4 = 0
\end{aligned}
\end{equation}
In order point $x = (x_1,x_2,x_4) = (0,\frac{\ell_1-\ell_3}{\ell_1+\ell_2-\ell_3},\frac{\ell_2}{\ell_1+\ell_2-\ell_3})$ to be optimal, we need to find a feasible dual solution such that point $x$ and the dual solution satisfy the Karush-Kuhn-Tucker (KKT) conditions together. In the primal program, point $x$ satisfies three constraints with equality, therefore, the corresponding dual variables $v_1, v_2,s_1$ can take nonzero values. The remaining dual variables are zero to satisfy the complementary slackness conditions. It is not difficult to see that feasible $v_1, v_2,s_1$ values for the dual program can be found if the following condition holds.
\begin{equation}\label{optimality_cond}
\begin{aligned}
    \frac{\ell_1(2\ell_2-\ell_3)}{\ell_1+\ell_2-\ell_3} \geq \ell_{3,0}
\end{aligned}
\end{equation}
We next show that when $\frac{\gamma}{d^a} > 3^a$ and $\ell_3 < \ell_1$, this condition is satisfied. Towards this end, we first write the condition in \eqref{optimality_cond} through using the link capacities in \eqref{modified_linkcap} and obtain the following inequality. 
\begin{multline}\label{modified_opt_cond}
    f = \log\left(\frac{\gamma}{\beta^ad^a}\right)\log\left(\frac{\gamma(1-2\beta)^a}{(1-\beta)^{2a}d^a}\right) \\- \log\left(\frac{\gamma(1-2\beta)^a}{(1-\beta)^a\beta^ad^a}\right)\log\left(\frac{\gamma}{d^a}\right) \geq 0
\end{multline}
We should note that the condition $\ell_3 < \ell_1$ can be equivalently written as $\beta < \frac{1}{3}$. We then find a condition on $\log(s)$ such that the condition in \eqref{modified_opt_cond} is satisfied where $s = \gamma/d^a$.
\begin{equation}
\begin{aligned}
    log(s) \geq a\frac{log(\beta)}{log(1-\beta)}\log\left(\frac{(1-\beta)^2}{1-2\beta}\right)
\end{aligned}
\end{equation}
We can denote the function on the RHS of the inequality as $\hat{f}$. Through using basic calculus, it is not difficult to see that the derivative of the function $\hat{f}$ is non-negative when $\beta < \frac{1}{3}$. Therefore, the function $\hat{f}$ is a monotonically increasing function in $\beta$. At $\beta = \frac{1}{3}$, $\hat{f} \approx 1.13a$. Hence, $s$ should be greater than or equal to $2^{1.13a}$. Since we assume that $s > 3^a$, this condition is satisfied. As a result, point $x =  (0,\frac{\ell_1-\ell_3}{\ell_1+\ell_2-\ell_3},\frac{\ell_2}{\ell_1+\ell_2-\ell_3})$ is one of the optimal solutions and the corresponding $\lambda_{\ell_{3,0}} = 0$. In order to show that $\lambda_1+\lambda_2 = 1$, we can use the same argument as in the proof of $\lambda_2+\lambda_{\ell_{2,1}} = 1$. We can start with an optimal solution such that $\lambda_1+\lambda_2$ < 1 and $\lambda_{\ell_{3,0}} = 0$, then we can increase $\lambda_1$ until the inequality is satisfied with equality. If we use the same flow variables as in the optimal solution, all constraints will be satisfied and we reach the same objective value. This concludes the proof of the equality in \eqref{property3} and Lemma \ref{properties}.

\section{Proof of Lemma~\ref{third_range}}\label{appendix_third_range}
Here, we prove Lemma~\ref{third_range} by showing that when $0 < \beta \leq d/\gamma^{1/a}$, then for any fixed $\lambda_2 \in [0,1]$ either the RHS of $({\rm P4}d)$ or the RHS of $({\rm P4}f)$ is smaller than $\widebar{\msf{C}}^\star$. 

We start off by finding a condition on $\lambda_2$ such that the RHS of $({\rm P4}d)$ becomes greater than or equal to $\widebar{\msf{C}}^\star$.
\begin{align}\label{third_range_firstbound}
&\log\left(\frac{\gamma}{d^a}\right)-\lambda_2\log\left((1-\beta)^a\right)-(1-\lambda_2)\log\left(\beta^a\right) \geq \widebar{\msf{C}}^\star \nonumber \\
&\implies \lambda_2\left(\log\left(\beta^a\right)-\log\left((1-\beta)^a\right)\right) \geq \log\left(\frac{\gamma}{d^a}\right)+\log\left(\beta^a\right) \nonumber \\
&\implies \lambda_2 \leq \frac{\log\left(\frac{\gamma}{d^a}\right)+\log\left(\beta^a \right)}{\log\left(\left(\frac{\beta}{1-\beta}\right)^a\right)}
\end{align}
When the condition on $\lambda_2$ in \eqref{third_range_firstbound} holds, the first bound becomes greater than or equal to $\widebar{\msf{C}}^\star$. 

We next observe that in order for the RHS of $({\rm P4}f)$ to be greater than or equal to $\widebar{\msf{C}}^\star$, we get the following condition on $\lambda_2$.
\begin{align}\label{third_range_thirdbound}
&(1-\lambda_2)\left(\log\left(\frac{\gamma}{d^a}\right)-\log\left(\left(1-2\beta\right)^a\right)\right) \nonumber
\\
&\qquad +2\lambda_2\left(\log\left(\frac{\gamma}{d^a}\right)-\log\left(\left(1-\beta\right)^a\right)\right) \geq \widebar{\msf{C}}^\star \nonumber 
\\
&\implies \lambda_2\left(\log\left(\frac{\gamma}{d^a}\right)-2\log\left(\left(1-\beta\right)^a\right)+\log\left(\left(1-2\beta\right)^a\right)\right) \nonumber\\
&\qquad\qquad -\log\left(\left(1-2\beta\right)^a\right) \geq \log\left(\frac{\gamma}{d^a}\right) \nonumber
\\
&\implies \lambda_2 \geq \frac{\log\left(\frac{\gamma}{d^a}\right)+\log\left(\left(1-2\beta\right)^a\right)}{\log\left(\frac{\gamma\left(1-2\beta\right)^a}{d^a(1-\beta)^{2a}}\right)}.
\end{align}
We will now show that for all $0 < \beta \leq d/\gamma^{1/a}$ and assuming that $\frac{\gamma}{d^a} > 3^a$, both inequalities ~\eqref{third_range_firstbound} and~\eqref{third_range_thirdbound} cannot be satisfied simultaneously by showing that the RHS of~\eqref{third_range_thirdbound} is greater than the RHS of~\eqref{third_range_firstbound} for this range of $\beta$, thus proving our lemma. 
% Using these two bounds on $\lambda_2$ in \eqref{third_range_thirdbound} is greater than the upper bound in \eqref{third_range_firstbound}
% If the lower bound on $\lambda_2$ in \eqref{third_range_thirdbound} is greater than the upper bound in \eqref{third_range_firstbound}, there is no feasible $\lambda_2$ that makes both of these bounds greater than or equal to $C^\star$. In this case, the capacity achieved becomes smaller than $C^\star$ since either the first or the third bound becomes smaller than $C^\star$. Therefore, we next show that the lower bound in \eqref{third_range_thirdbound} is greater than the upper bound in \eqref{third_range_firstbound} under the condition in Lemma \ref{third_range} and $\frac{\gamma}{d^a} \geq 3^a$.
We would like to show that
\begin{equation}
\begin{aligned}
\frac{\log\left(\frac{\gamma(1-2\beta)^a}{d^a}\right)}{\log\left(\frac{\gamma(1-2\beta)^a}{d^a(1-\beta)^{2a}}\right)} > \frac{\log\left(\frac{\gamma\beta^a}{d^a}\right)}{\log\left(\left(\frac{\beta}{1-\beta}\right)^a\right)}
\end{aligned}
\end{equation}
This condition can be rewritten by organizing terms as showing that $0 < \beta \leq d/\gamma^{1/a}$, we have that
% \begin{multline*}
% \log\left(\frac{\gamma(1-2\beta)^a}{d^a}\right)\log\left(\left(\frac{\beta}{1-\beta}\right)^a\right) \\< \log\left(\frac{\gamma\beta^a}{d^a}\right)\log\left(\frac{\gamma(1-2\beta)^a}{d^a(1-\beta)^{2a}}\right)
% \end{multline*}
\begin{multline}\label{third_range_condition}
    f(\beta) = \log\left((1-\beta)^{2a}\right)\log\left(\frac{\gamma\beta^a}{d^a}\right)\\-\log\left(\frac{\gamma(1-2\beta)^a}{d^a}\right)\log\left(\frac{\gamma(1-\beta)^a}{d^a}\right) < 0.
\end{multline}
Note that at $\beta = d/\gamma^{1/a}$, $f(\beta) <0$ for $\gamma/d^a > 3^a$. Thus, it is sufficient to show that the function is monotonically increasing over our range of interest. Taking the derivative, we get that
\begin{multline}\label{eq:derivative_f_beta}
    \frac{df(\beta)}{d\beta} = \frac{-2}{1-\beta}\log\left(s\beta^a\right)+\frac{2}{\beta}\log\left((1-\beta)^a\right)\\+\frac{2}{1-2\beta}\log\left(s(1-\beta)^a\right)+\frac{1}{1-\beta}\log\left(s(1-2\beta)^a\right), 
\end{multline}
where $s = \gamma/d^a$. If we show that the derivative is positive, then we are done as the function is by consequence monotonically increasing.

Note that, when $\frac{\gamma}{d^a} > 3^a$ and $0 < \beta \leq d/\gamma^{1/a}$, it is not difficult to see that the only negative term in~\eqref{eq:derivative_f_beta} is $\frac{2}{\beta}\log\left((1-\beta)^a\right)$. Therefore, we next find a condition on $s$ that makes $\frac{2}{1-2\beta}\log\left(s(1-\beta)^a\right) + \frac{2}{\beta}\log\left((1-\beta)^a\right)> 0$. This can be rewritten as aiming to show that
\[
\log(s) > \log\left((1-\beta)^a\right)\frac{\beta-1}{\beta}.
\]
It is not difficult to verify through basic calculus that given that $\frac{\gamma}{d^a} > 3^a$, then the RHS above is monotonically decreasing and is less than $\log(s)$ at $\beta = 0$.
Thus, we have shown that constraints~\eqref{third_range_firstbound} and~\eqref{third_range_thirdbound} cannot be satisfied simultaneously and thus concluding the proof of Lemma~\ref{third_range}.

\end{document}